\documentclass[11pt]{article}
\usepackage[letterpaper]{geometry}
\usepackage[parfill]{parskip}
\usepackage{amsmath,amsthm,amssymb,bbm}
\usepackage{mathtools}
\usepackage{cases}
\usepackage{graphicx}
\usepackage{tabularx}
\usepackage{microtype}
\usepackage{enumitem}

\usepackage{authblk}

\usepackage{url}
\usepackage[colorlinks,citecolor=blue,urlcolor=blue,linkcolor=blue,linktocpage=true]{hyperref}
\usepackage{cleveref}
\crefformat{equation}{(#2#1#3)}
\crefrangeformat{equation}{(#3#1#4) to~(#5#2#6)}
\crefname{equation}{}{}
\Crefname{equation}{}{}

\usepackage[authoryear]{natbib}

\newtheoremstyle{mythmstyle}
  {8 pt} 
  {3 pt} 
  {} 
  {} 
  {\bfseries} 
  {.} 
  {.5em} 
  {} 

\theoremstyle{mythmstyle}

\newtheorem{theorem}{Theorem}
\newtheorem{lemma}[theorem]{Lemma}

\newtheorem*{definition*}{Definition}
\newtheorem*{remark*}{Remark}
\crefname{definition}{\textbf{definition}}{definitions}
\Crefname{definition}{Definition}{Definitions}
\crefname{assumption}{\textbf{assumption}}{assumptions}
\Crefname{assumption}{Assumption}{Assumptions}

\DeclarePairedDelimiter\norm{\lVert}{\rVert}
\makeatletter
\newcommand{\opnorm}{\@ifstar\@opnorms\@opnorm}
\newcommand{\@opnorms}[1]{%
  \left|\mkern-1.5mu\left|\mkern-1.5mu\left|
   #1
  \right|\mkern-1.5mu\right|\mkern-1.5mu\right|
}
\newcommand{\@opnorm}[2][]{%
  \mathopen{#1|\mkern-1.5mu#1|\mkern-1.5mu#1|}
  #2
  \mathclose{#1|\mkern-1.5mu#1|\mkern-1.5mu#1|}
}
\makeatother

\DeclareMathOperator{\sign}{sign}

\def\Real{\mathbb{R}}

\DeclareMathOperator{\cov}{Cov}

\def\Proj{\mathcal{P}}

\newcommand{\argmax}{\operatornamewithlimits{arg\,max}}

\newcommand{\ConvexIndicator}[1]{\mathbb{I}_{#1}}

\DeclareMathOperator{\trace}{trace}

\DeclarePairedDelimiterX\innerp[2]{\langle}{\rangle}{#1,#2}

\newcommand{\Fantope}[2][]{\mathcal{F}_{#1}^{#2}}
\newcommand{\D}[2][]{\mathcal{D}_{#2}^{#1}}

\newcommand{\inner}[2]{\langle #1,#2\rangle}

\newcommand{\soft}[1]{\mathcal{S}_{#1}}
\newcommand{\InputMatrix}{S}
\newcommand{\PopMatrix}{\Sigma}
\newcommand{\PopProjection}{\Pi}
\newcommand{\ndim}{k}
\newcommand{\PrimalVariable}{H}

\newcommand{\DualVariable}{Z}

\newcommand{\PopOperator}{\Gamma}

\newcommand{\EstProjection}{\widehat{\PopProjection}}

\newcommand{\InputError}{W}

\newcommand{\SparsePenalty}{{\rho_2}}
\newcommand{\SmoothPenalty}{{\rho_1}}
\newcommand{\DiffOperator}{D}
\newcommand{\DiffOperatorHalf}{\Delta}
\newcommand{\Gap}{\delta}
\newcommand{\step}{\tau}

\usepackage{algpseudocode,algorithm}

\begin{document}

\title{Localized Functional Principal Component Analysis}

\author[1]{Kehui Chen}
 \author[2]{Jing Lei}
\affil[1]{Department of Statistics and Department of Psychiatry, University of Pittsburgh}
\affil[2]{Department of Statistics, Carnegie Mellon University}

\maketitle

\begin{abstract}
We propose localized functional principal component analysis (LFPCA), looking for orthogonal basis functions with localized support regions that explain most of the variability of a random process. The LFPCA is formulated as a convex optimization problem through a novel Deflated Fantope Localization method and is implemented through an efficient algorithm to obtain the global optimum. We prove that the proposed LFPCA converges to the original FPCA when the tuning parameters are chosen appropriately. Simulation shows that the proposed LFPCA with tuning parameters chosen by cross validation can almost perfectly recover the true eigenfunctions and significantly improve the estimation accuracy when the eigenfunctions are truly supported on some subdomains. In the scenario that the original eigenfunctions are not localized, the proposed LFPCA also serves as a nice tool in finding orthogonal basis functions that balance between interpretability and the capability of explaining variability of the data. The analyses of a country mortality data and a growth curve data reveal interesting features that cannot be found by standard FPCA methods.
\end{abstract}


\section{Introduction}\label{sec:introduction}
Functional principal component analysis has emerged as a major tool to explore the source of variability in a sample of random
curves and has found wide applications in functional regression, curve classification, and clustering \citep{castro1986principal,rice1991estimating,cardot2000nonparametric,yao2005functional,ramsay2005functional,hall2006properties}. In this paper, we consider functional principal component analysis with localized support regions.
That is, for a smooth random function $X$, we look for orthogonal basis functions with localized support regions that explain most of the variance. The main motivation of the localized functional principal component analysis (LFPCA) is to find a parsimonious linear representation of the data that balances the interpretability and the capability of explaining variance of the stochastic process. 

The proposed method outputs localized basis functions whose localization level is controlled by a localization tuning parameter. We propose two methods to select the localization parameter, corresponding to two useful applications of our proposed method. First, one can choose the localization parameter by maximizing the explained variance of the random process computed by $V$-fold cross-validation. Our simulation shows that when the eigenfunctions truly have localized support regions, the proposed LFPCA with a localization parameter chosen by cross-validation significantly improves the estimation accuracy of the eigenfunctions compared to standard FPCA methods. On the other hand, when the original eigenfunctions are not localized, the localization parameter chosen by cross-validation is expected to be very close to zero and it is confirmed by our numerical studies that the performance of the proposed LFPCA is almost identical to standard FPCA methods. 
The second method of choosing the localization parameter is to seek the most localized basis functions that explain a fixed level of variance. This method is particularly useful when the standard eigenfunctions are not localized, and it makes sense for the proposed LFPCA not to target at the standard eigenfunctions but to balance between interpretability and the capability of explaining variance of the stochastic process. Details can be found in Section 3.2. We consider two data applications to illustrate the second method of choosing the localization parameter. One is a country mortality data, where the mortality rates at age 60 were recorded from year 1960--2006 for 27 countries around the world. The first three localized basis functions, explaining more than 85\% variance, correspond to variational modes around mid 1990s, 1980s, and 1960s, respectively. Another example is a growth curve data, where the height of 54 girls were densely recorded from year 1 to year 18. The first two localized basis functions explain more than 85\% variance, and clearly indicate that the main variational modes of female height growth are around age 12 and around age 5, perfectly matching the knowledge of growth spurts. These interesting features cannot be revealed by standard FPCA.

Domain localization has been studied by several authors in the functional regression model: $E(Y|X) = a + \int_{\mathcal{T}} X(t)\beta(t) dt$, where the coefficient function $\beta(t)$ is desired to be zero outside a subdomain $\mathcal{T}_0\in \mathcal{T}$ with the purpose of improved interpretability \citep{james2009functional,zhao2012wavelet,zhou2013functional}. Most of these methods turn the problem into a variable selection problem and use LASSO type penalties. In a recent thesis work \cite{lin2013some}, interpretable functional principal component analysis is studied, which has a very similar flavor as our proposed LFPCA. In their work, an $\ell_0$ penalty is added on eigenfunctions and a greedy algorithm based on basis expansion of curves is developed to solve the non-convex optimization problem. The formulation of localization or domain selection in the context of functional principal component analysis is quite challenging for at least two reasons. First, the eigen problem together with a localization penalty is usually not convex, and in general it is an NP-hard problem to find a global optimum.
Second, in order to obtain a sequence of mutually orthogonal eigen-components, a commonly taken procedure is to deflate the empirical covariance operator at step $j$ by removing the effect of the previous $j-1$ components \citep{white1958computation,Mackey08}. But with the localization penalty in the objective function, this procedure can not guarantee the orthogonality of such sequentially obtained eigen-components. In sequential estimation of principal components, being orthogonal to the first component is a natural requirement when looking for the second component, otherwise the maximization over second direction is not well-defined since the solution would still be the first direction. From a dimension reduction perspective, the orthogonality is also appealing since the resulting $k$ dimensional orthogonal basis leads to very simple calculation for subsequent inferences. 

The main contribution of this paper is three-fold. First, we formulate the LFPCA as a convex optimization problem with explicit constraint on the orthogonality of eigen-components. Second, we provide an efficient algorithm to obtain the global maximum of this convex problem. Third, we carefully investigate the estimation error from the discretized data version to the functional continuous version, as well as the complex interaction between the eigen problem and the localization penalty, and prove consistency of the estimated eigenfunctions. The starting point of our method is a sup-norm consistent estimator of the covariance operator, up to a constant shift on the diagonal. For dense and equally spaced observations with or without measurement error, the proposed method can be directly carried out on the sample covariance, i.e., without the need to perform basis expansion, smoothing of the individual curves, or smoothing of the estimated covariance operator. For other designs of functional data, the proposed method is still applicable when an appropriate covariance estimator is available. 

Our formulation of LFPCA borrows ideas from recent developments in sparse principal component analysis. In \cite{vu2013fantope}; \cite{lei2014sparsistency}, a similar convex framework based on Fantope Projection and Selection has been proposed to estimate a $k$ dimensional sparse principal subspace of a high dimensional random vector (see also \cite{AEJL:2007} for $k=1$).  These sparse subspace methods are useful when the union of the support regions of several leading eigenvectors is sparse. In sparse PCA settings, the notion of sparsity requires the proportion of non-zero entries in the leading eigenvectors to vanish as the dimensionality increases and therefore it makes sense to consider the union of the support regions of several leading eigenvectors to be sparse. However, in functional data settings, the length ratio of a support subdomain over the entire domain is determined by
the random curve model and usually a constant, and the union of several leading subdomains can be as large as the entire domain.
This is also the reason that we use the notion ``localized'' instead of ``sparse''. It has remained challenging to obtain sparse eigenvectors sequentially that each one is allowed to have a different support region. A particular challenge is the interaction between orthogonality and the sparse penalty. 
Besides the difference between functional PCA and sparse PCA, one main extension developed in our method is the construction of a deflated Fantope to estimate individual eigen-components sequentially, with possibly different support regions and guaranteed orthogonality. This deflated Fantope formulation is of independent interest in many other structured principal component analysis.

The rest of this paper is organized as follows. In section 2, we introduce the formulation of localized functional principal component analysis. Section 3 derives the solution to the optimization problem and describes the algorithm as well as the selection of tuning parameters. Section 4 contains the consistency results. Section 5 and Section 6 present numerical experiments and data examples to illustrate our method. Section 7 contains some discussions and extensions.  Technical details are included in the Appendix.

\section{LFPCA Through Deflated Fantope Localization}\label{sec:formulation}
We consider a square integrable random process $X(t):\mathcal T\mapsto \Real$ over a 
compact interval $\mathcal{T}\subset \Real$, with mean and covariance functions
$\mu(t) = E X(t) $ and $\PopOperator(s,t)= \cov(X(s), X(t))$, and covariance operator
$(\PopOperator f)(t)  = \int_{s\in \mathcal{T}}f(s) \PopOperator(t,s)ds.$
Under the minimal assumption that $\PopOperator(s,t)$ is continuous over  $(s,t)$, this operator $\PopOperator$ has orthonormal eigenfunctions $\phi_j(t), j=1,2,\cdots,$ with nonincreasing eigenvalues $\lambda_j$, satisfying
$\PopOperator\phi_j = \lambda_j \phi_j.$
The well known Karhunen-Lo\`{e}ve expansion then gives the representation
\begin{equation}
\label{KLexpansion}
X(t)=\mu(t)+\sum_{j=1}^{\infty}\xi_{j}\phi_j(t),
\end{equation}
where $\xi_j, \ j\geq 1,$ is a sequence of uncorrelated random variables satisfying $E(\xi_j)=0$
and ${\rm var}(\xi_j)=\lambda_j$, with an explicit representation
$\xi_j= \int_{t\in \mathcal{T}}(X(t)-\mu(t))\phi_j(t)dt.$ A key inference task in functional principal component analysis (FPCA) is to estimate the leading
eigenfunctions $\phi_j(t),~1\leq j \leq k,$ from $n$ independent sample curves.

In practice, the underlying sample curves $X_i(t),~1\le i\le n,$ are usually recorded at a grid of points and may
be contaminated with additive measurement errors. We start with dense and equally spaced observations
\begin{equation}\label{eq:obs-model}
Y_{il} = X_i(t_l) + \epsilon_{il} = \mu(t_l)+\sum_{j=1}^{\infty}\xi_{ij}\phi_j(t_l)+\epsilon_{il}\,,~~i = 1\,,\dots\,, n\,,~~ l= 1\,,\dots\,, p\,,
\end{equation} where $\epsilon_{il}$ are independent noises
with mean zero and variance $\sigma^2$, and
$t_l,~1\le l\le p,$ are grid points in $\mathcal T$ at which observations are recorded. Starting from such discrete and possibly noisy data, there are different ways of introducing smoothness in the estimation of FPCA. \cite{rice1991estimating}; \cite{silverman1996smoothed}; \cite{huang2008functional}, among others, studied approaches where a smoothness penalty on eigenfunctions is integrated in the optimization step of eigen-decomposition.
Let $\InputMatrix$ be the $p\times p$ sample covariance matrix of the observed vector $Y$, and $v$ be a $p$ dimensional vector. \cite{rice1991estimating} used a roughening matrix $\DiffOperator=\DiffOperatorHalf^T \DiffOperatorHalf$ where $\DiffOperatorHalf\in
\Real^{(p-2)\times p}$ is a second-differencing operator:
$$
\DiffOperatorHalf_{ij} =\left\{ \begin{array}{ll}
  1, &~\text{if}~j\in\{i,i+2\}\,,\\
  -2, &~\text{if}~j=i+1\,,\\
  0,&~\text{otherwise}\,.
\end{array}\right.
$$
A smoothed eigenfunction estimator is obtained by solving the following
eigen problem:
\begin{equation}\max v^T (\InputMatrix- \SmoothPenalty D)v,  {\rm\ \ s.t.\ } 
  \norm{v}_2 = 1,\end{equation} where $\SmoothPenalty$ is a smoothing parameter, and 
  $\norm{v}_2$ is the Euclidean norm of $v$.

A straightforward approach to localize the estimated eigenfunctions is to add another localization penalty:
\begin{equation}\label{eq:pen-pca-nonconv}
  \max v^T (\InputMatrix- \SmoothPenalty D)v - \SparsePenalty \norm{v}_1, {\rm\ \ s.t.\ } \norm{v}_2 = 1,\end{equation}
where $\SparsePenalty$ is a tuning parameter, and $\norm{v}_1$ is the $\ell_1$ norm of $v$.
However, this is not a convex problem and there are
no known algorithms that can efficiently find a global optimum even for the first eigen-component.

Here we propose a novel sequential estimation procedure based on the idea
of estimating the rank one projection matrix $vv^T$. Let 
$\innerp{A}{B}=\trace(A^TB)$ for matrices $A$, $B$ of compatible dimensions. Denote 
$\norm{H}_{1,1}$ the matrix $\ell_1$ norm, which is the sum absolute value of all 
entries in $H$. Starting from $\EstProjection_0 = 0$, for each $j=1,...,k$, the 
$j$th localized eigen-component is estimated as follows.
 \begin{equation}\label{opt2}\begin{split} & H_j = \argmax 
   \innerp{\InputMatrix-\SmoothPenalty D}{\PrimalVariable} - \SparsePenalty 
   \norm{\PrimalVariable}_{1,1}, {\rm \ s.t. \ } \PrimalVariable \in 
   \D{\EstProjection_{j-1}},\\
 & \hat v_j = {\rm\ the\ first\ eigenvector\ of\ } H_j,\\
  & \EstProjection_j =  \EstProjection_{j-1} + \hat v_j \hat 
  v_j^T,\end{split}\end{equation} where,
  for any $p\times p$ projection matrix $\Pi$, $$\D{\Pi} \coloneqq 
  \{\PrimalVariable: 0 \preceq \PrimalVariable \preceq I, ~\trace(\PrimalVariable) = 
  1, \ {\rm and}\ \innerp{H}{\Pi} =0 \},$$
  and, for symmetric matrices $A$, $B$, ``$A\preceq B$'' means that $B-A$ is positive semidefinite.

  Problem (\ref{opt2}) is a convex relaxation of (\ref{eq:pen-pca-nonconv}) with integrated orthogonality constraints on the estimated localized eigen-components.  With the estimated $\hat v_j$, we can easily obtain an
estimate $\hat\phi_j(t)$ of the localized eigenfunction $\phi_j(t)$ by standard interpolation techniques such as linear interpolation, plus an optional final step of re-orthogonalization and re-normalization.

With appropriately chosen tuning parameters, the performance of the proposed method ties to the maximum entry-wise error of the discretized covariance estimator $S$ (see
\Cref{sec:theory} for details).
Our presentation will focus on a sample covariance $S$ that is computed from dense and equally spaced observations, but the proposed method is not restricted to a dense regular design as long as a reasonable covariance estimate can be obtained. More discussions can be found in Section 7.

To solve for $\hat v_j$ and $\hat \phi_j(t)$, the key step in (\ref{opt2}) is to solve
\begin{equation} \max_H \innerp{\InputMatrix-\SmoothPenalty D}{\PrimalVariable} - \SparsePenalty \norm{\PrimalVariable}_{1,1}, \ \text{s.t.} \ \PrimalVariable \in \D{\Pi}\label{eq:H_i},\end{equation} where $\Pi = \hat\Pi_{j-1}$ at step $j$.  In next
  section we present an algorithm that solves problem \eqref{eq:H_i}, with a discussion
  on the choice of tuning parameters $\SmoothPenalty$ and $\SparsePenalty$.

In the sparse PCA literature, \cite{AEJL:2007}; \cite{vu2013fantope} have considered the following problem,
\begin{align}
  \max_H \innerp{\InputMatrix}{H}-\rho \norm{H}_{1,1}, {\rm \ s.t. \ }  H\in \Fantope{d}\,,\label{eq:pen-pca-conv}
\end{align}
where the convex set $\Fantope{d} \coloneqq \{\PrimalVariable: 0 \preceq \PrimalVariable \preceq I,~ \trace(\PrimalVariable) = d\}$ is called the  \emph{Fantope} of degree $d$ \citep{Dattorro}, and
 is the convex hull of all rank $d$ projection matrices. 
While the convex relaxation given in \eqref{eq:pen-pca-conv} allows us
to estimate $d$-dimensional sparse principal subspaces, it does not lead to
a sequence of mutually orthogonal eigenvectors with different support regions. 

To ensure orthogonality among the estimated eigenvectors, we consider $\D{\Pi} \coloneqq \{\PrimalVariable: \PrimalVariable \in \Fantope{1}, \ {\rm and}\ \innerp{H}{\Pi} =0 \}$, which we call the \emph{deflated Fantope}. It can be naturally generalized to $\mathcal{D}_{\Pi}^{d} \coloneqq \{\PrimalVariable: \PrimalVariable \in \Fantope{d}, \ {\rm and}\ \innerp{H}{\Pi} =0 \}$ to estimate mutually orthogonal principal subspaces.  Such a
feasibility deflation technique is quite different from the commonly suggested
matrix deflation techniques in sequential estimation of eigenvectors (see \cite{Mackey08} for example).

\section{Algorithm}\label{sec:algorithm}
\subsection{Deflated Fantope Localization using ADMM}
The main difficulty in solving problem \eqref{eq:H_i} is the complex interaction between
the $\ell_1$ penalty and the deflated Fantope constraint.  To overcome this difficulty, we
write \eqref{eq:H_i} in an equivalent form to separate the $\ell_1$ penalty and deflated Fantope constraint:
\begin{equation}\label{eq:H_i-equiv}
\begin{aligned}  \min_{\PrimalVariable,\DualVariable} & 
~\ConvexIndicator{\D{\Pi}}(\PrimalVariable)-
\innerp{S-\SmoothPenalty\DiffOperator}{\PrimalVariable} + \SparsePenalty 
\norm{\DualVariable}_{1,1}\,,\\
  \text{s.t.} &~ \PrimalVariable - \DualVariable = 0 \,,
  \end{aligned}
\end{equation}
where $\ConvexIndicator{\D{\Pi}}$ is the convex indicator function, which
is $\infty$ outside $\D{\Pi}$ and 0 inside $\D{\Pi}$.
Problem \eqref{eq:H_i-equiv} is a convex global variable consensus optimization, which can be solved using alternating direction method of multipliers (ADMM, \cite{Boyd-ADMM}).  We describe in \Cref{alg:admm} an ADMM algorithm that solves \eqref{eq:H_i-equiv} and hence \eqref{eq:H_i}. It extends the FPS algorithm
in \cite{vu2013fantope} to the deflated Fantope.
\begin{algorithm}[tb]
\caption{Deflated Fantope Localization using ADMM}
\label{alg:admm}
\begin{algorithmic}
    \Require{$\InputMatrix = \InputMatrix^T$, $\Pi$, $\DiffOperator$, $\SmoothPenalty,\SparsePenalty \geq 0$, $\step > 0$, $\epsilon > 0$}
    \State $\DualVariable^{(0)} \gets 0, W^{(0)} \gets 0$
    \Comment{Initialization}
    \Repeat \quad $r = 1,2,\ldots$
    \State
        $\PrimalVariable^{(r)} \gets \Proj_{\D{\Pi}} \big[\DualVariable^{(r-1)} - 
        W^{(r-1)} + (\InputMatrix - \SmoothPenalty\DiffOperator)/\step \big]$
        \Comment{Deflated Fantope projection}
    \State
        $\DualVariable^{(r)} \gets \soft{\SparsePenalty/\step} \big( \PrimalVariable^{(r)} + W^{(r-1)} \big)$
        \Comment{Elementwise soft thresholding}
    \State
        $W^{(r)} \gets W^{(r-1)} + \PrimalVariable^{(r)} - \DualVariable^{(r)}$
        \Comment{Dual variable update}
    \Until{$\norm{\PrimalVariable^{(r)} - \DualVariable^{(r)}}_F^2 \lor \step^2 \norm{\DualVariable^{(r)} - \DualVariable^{(r-1)}}_F^2 \leq \epsilon^2$}
        \Comment{Stopping criterion}
    \\
    \Return{$\DualVariable^{(r)}$}
\end{algorithmic}
\end{algorithm}
\vskip 5pt

The two matrix operators used in the algorithm are defined as follows.

\noindent (i) Soft-thresholding operator: for any $a>0$,
$$
\soft{a}(x) = \sign(x)\max(|x|-a,0)\,.
$$

\noindent (ii) Deflated-Fantope-projection operator:
For any $p\times p$ symmetric matrix $A$ and projection matrix $\Pi$,
$$
\Proj_{\D{\Pi}}(A)\coloneqq \arg\min_{B\in\D{\Pi}}\norm{A-B}_F^2
$$
is the Frobenius norm projection of $A$ onto the deflated Fantope $\D{\Pi}$.

A non-trivial subroutine in \Cref{alg:admm} is to calculate the deflated-Fantope-projection $\Proj_{\D{\Pi}}(A)$
for a symmetric matrix $A$.
The following lemma gives a close-form characterization of the deflated-Fantope-projection operator.
\vskip 5pt
\begin{lemma}\label{lem:char-dfantope}
Let $\Pi=VV^T$, where $V$ is a $p\times d$ matrix with orthonormal columns.
Let $U$ be a $p\times (p-d)$ matrix that forms an
orthogonal complement basis of $V$.  Then
  $$
  \Proj_{\D{\Pi}}(A) =
  U\left[\sum_{i=1}^{p-d} \gamma_i^+(\theta) \eta_i \eta_i^T\right]U^T,
  $$
  where $(\gamma_i,\eta_i)_{i=1}^{p-d}$ are eigenvalue-eigenvector pairs of $U^T A U$: $
  U^T A U = \sum_{i=1}^{p-d} \gamma_i \eta_i \eta_i^T\,,
  $
  and $\gamma_i^+(\theta)=\min(\max(\gamma_i-\theta,0),1)$, with $\theta$ chosen such
  that $\sum_{i=1}^{p-d}\gamma_i^+(\theta)=1$.
\end{lemma}
\vskip 5pt

The next theorem  ensures the convergence of our algorithm to a global optimum of problem \eqref{eq:H_i}. The proofs of \Cref{lem:char-dfantope} and \Cref{thm:convergence} are deferred to the Appendix.

\vspace{5 pt}

\begin{theorem}\label{thm:convergence}
  In \Cref{alg:admm}, $\DualVariable^{(r)}\rightarrow \PrimalVariable^*$, $\PrimalVariable^{(r)}\rightarrow \PrimalVariable^*$ as $r\rightarrow\infty$, where $\PrimalVariable^*$ is a global optimum of problem
  \eqref{eq:H_i}.
\end{theorem}

In the proof of \Cref{thm:convergence} we will see that the auxiliary number $\step$ used in
\Cref{alg:admm} plays a role that is similar to the step size commonly seen in iterative
convex optimization solvers.  The particular choice of $\step$ does not affect the
theoretical convergence of
the ADMM algorithm. There are some general guidelines on the practical choice of $\step$ and
we refer to \cite{Boyd-ADMM} for further details.

\subsection{Choice of Tuning Parameters}
The optimization problem (\ref{opt2}) involves two tuning parameters:
$\rho_1$ controls the roughness of eigenfunctions; and $\rho_2$ controls the localization of the eigenfunctions.
We present a two-step approach where
$\SmoothPenalty$ is chosen first and kept the same for all $1\leq j\leq k$, and then $\SparsePenalty$ is determined sequentially for each eigenfunction $\phi_j(t)$ and denoted by $\rho_{2,j}$. The two parameters can be chosen together by straightforward modification but computationally it would be a bit intensive. 

The choice of the smoothing parameter $\rho_1$ has been discussed in \cite{rice1991estimating}, and they recommended cross-validation or manual selection. In our simulation and data analysis, we have used $V$-fold cross-validation. First the data is divided into $V$ folds, denoted by $\mathcal{P}_1$, $\mathcal{P}_2$, \dots, $\mathcal{P}_V$. Let $H^{(-v)}_j(\rho_1, \rho_2)$ be the estimated $H_j$ in (\ref{opt2}) using data other than $\mathcal{P}_v$ with tuning parameters $\rho_1$ and $\rho_2$. Let $S^{v}$ be the discrete covariance estimated from data $\mathcal{P}_v$. The smoothing parameter is chosen by maximizing the cross-validated inner product of $H_1^{(-v)}$ and $S^{v}$:
\begin{align}
\hat \rho_1 = &\argmax_{\rho\in\mathcal A_1} \sum_{v=1}^{V}\innerp{H^{(-v)}_1(\rho, 0)}{S^{v}}, \label{hatrho1}
\end{align}
where $\mathcal A_1$ is a candidate set of $\rho_1$ and empirically we found that a sequence between 0 and $ p {\rm \ times\ the\ largest\ eigenvalue\ of\ }\InputMatrix$ works well. 

In the following, we present two methods for the choice of $\rho_2$ given a pre-chosen smoothing parameter $\rho_1^*$. The first method is to choose $\rho_{2,j}$ by maximizing the cross-validated inner product of $H_j^{(-v)}$ and $S^{v}$:
\begin{align}
\hat \rho_{2,j} = &\argmax_{\rho\in\mathcal A_{2,j}} \sum_{v=1}^{V}\innerp{H_j^{(-v)}(\rho_1^*, \rho)}{S^{v}}\,,~j=1,2,...,\ndim\,,\label{hatrho2}
\end{align}
where $\mathcal A_{2,j}$ is a candidate set for $\rho_{2,j}$, and we propose to use a sequence between 0 and the 95\% quantile of absolute values of off-diagonal entries in $S_j$, with $S_j = (I-\hat\Pi_{j-1})S(I-\hat\Pi_{j-1})$. 


The $V$-fold cross-validation approach is expected to give a $\hat \rho_2$ that indicates the true localization level of the eigenfunctions.
The criterion in our proposed cross-validation corresponds to maximizing $\langle \Sigma, \hat H(\rho_2)\rangle$, where $\Sigma$ is the discretized population covariance and is substituted by the test sample covariance in practice.  When the true eigenvector $v$ is localized,
$\langle \Sigma, \hat H(\rho_2) \rangle$ shall be maximized at approximately the value of $\rho_2^*$ which corresponds to the ideal localization level of the eigenfunction, i.e., $\norm{\hat H(\rho_2^*)}_{1,1}\approx\norm{vv^T}_{1,1}$. We shall expect $\langle \Sigma, \hat H(\rho_2) \rangle$ as a function of $\rho_2$ to be (i) monotonically increasing on $[0,\rho_2^*]$ as the search area gradually expands to cover the true eigenvector, and (ii) monotonically decreasing on $[\rho_2^*,\infty)$
as the search area goes unnecessarily larger so that the estimation becomes more noisy.
When the true eigenvector $v$ is not localized, then we shall expect $\langle \Sigma, \hat H(\rho_2) \rangle$ to be monotonically decreasing as $\rho_2$ increases.
 The numerical study confirms the good performance of the cross-validation method. See \Cref{sec:simulation} for more details.

In some applications, we may not want to target the standard eigenfunction, but instead we may want to find orthogonal linear expansions that balance the interpretability (localization) and the capability of explaining the variance of the process. We therefore propose a second method of choosing $\rho_2$, which is based on the notion of fraction of variance explained (FVE). For a $p$-dimensional vector $v$ of unit length and a sup-norm consistent estimator $S$ of the covariance operator,
\begin{equation}
  \label{FVE}
  FVE(v) = v^TSv/totV(S),
\end{equation}
where $totV(S)$ is the sum of positive eigenvalues of $S$. We note that for dense and equally spaced observations with measurement error, the $FVE(v)$ defined above is not directly applicable to a sample covariance $S$ because the sample covariance is sup-norm consistent up to a shift $\sigma^2$ on the diagonal, where $\sigma^2$ is the error variance. To avoid serious bias by the nugget effect on the diagonal, one may use the eigenvalues from the smoothed covariance. Another practical way is to approximate $totV(S)$ by the sum of the first $M$ leading eigenvalues of $S$, for a finite number $M$. For a reasonable error level, the nugget effect bias $M \sigma^2$ is small compared to the sum of the first $M$ eigenvalues of $S$, which is of order $p$ \citep{kneip2011factor}, while the remaining true eigenvalues beyond $M$ are usually very small because the smoothness of $X(t)$ ensures fast decay of eigenvalues. 
In numerical experiments where $FVE$ is needed for determining the number of principal components to be included, we use $M=\min(20,p-2)$.

To sequentially select the sparsity parameter $\rho_2$ for the $j$th
eigenfunction. Suppose that we have estimated $\hat v_i$ for $1\le i\le j-1$.
Let $\hat v_{j}(\rho_1^*, \rho)$ be the solution of (\ref{opt2}) by using a fixed $\rho_1^*$ and $\rho_2 = \rho$, with $\hat \Pi_{j-1}$ being the projector of the subspace spanned by
$(\hat v_i:1\le i\le j-1)$. We can define 
\begin{equation}
  \label{rFVE}
 rFVE(\rho) =  \frac{FVE(\hat v_{j}(\rho_1^*, \rho))}{FVE(\hat v_{j}(\rho_1^*, 0))}
 =\frac{\hat v_j^T(\rho_1^*,\rho) S \hat v_j(\rho_1^*,\rho) }{\hat v_j^T(\rho_1^*,0) S \hat v_j(\rho_1^*,0) },
\end{equation} and choose $\rho_{2,j}$ as  
\begin{equation}
  \label{hatrho2_2}
 \max\{\rho\in \mathcal A_{2,j}: rFVE(\rho) \ge 1-a\},
\end{equation}
where $a \in [0,1)$ is the proportion of FVE that one chooses to sacrifice in return of localization.
For any $a\in[0,1)$, a $\rho$ satisfying \eqref{hatrho2_2} always exists, because $rFVE(\rho)  = 1$ for $\rho = 0$ and $rFVE(\rho) \in [0,1)$ for $\rho > 0$. 
Equation \eqref{rFVE} also suggests that $rFVE(\rho)$ can be calculated
without computing $totV(S)$. 
Although the first localized basis function explains less variance than the standard eigenfunction, the lost proportion is likely to be picked up by the second component, and we are still able to explain a large proportion of the total variance with a small number of components. We illustrate this method with real data analyses in Section 6. 

\section{Asymptotic Properties}\label{sec:theory}
In this section we establish the $\ell_2$ consistency of the proposed estimator in an asymptotic setting where both the sample size $n$ and the number of grid points $p$ increase.
We will first provide sufficient conditions on the tuning parameters $\rho_1$ and $\rho_2$
such that the LFPCA estimate is consistent.  Our second result provides further insights on how the localization penalty $\rho_2$ affects the rate of convergence.
We make the following assumptions.\\
\noindent \textbf{A1.} The input matrix $\InputMatrix$ in \eqref{opt2}
satisfies sup-norm consistency up to a constant shift on the diagonal: for some constant
$\alpha\ge 0$ and a sequence $e_n=o(1)$,
\begin{align*}&\max_{1\le l, l'\le p}|\InputMatrix(l,l')-\PopOperator(t_l,t_{l'})-\alpha \mathbf 1(l=l')|=O_P(e_n)\,,~~\text{as}~~(n,p)\rightarrow\infty\,.
\end{align*}

{\bf Remark}: Assumption (A1) puts a mild condition on the input matrix $\InputMatrix$
that can be satisfied by many standard estimators.  Consider functional data with dense and equally spaced observations. If $S$ is the sample covariance estimator from the raw data,
standard large deviation bounds such as Bernstein's inequality  (\cite{VanW96}, Chapter II) imply that Assumption (A1) holds with $e_n=\sqrt{\log p / n}$ if $\log p / n\rightarrow 0$
and the random curve $X(t)$ as well as the observation error in model \eqref{eq:obs-model} has sub-Gaussian tails; see also \cite{kneip2011factor}. In this case $\alpha=\sigma^2$, the noise variance. 
If a smoothed covariance estimator $S$ is used, the sup-norm rate can be $\sqrt{\log n/n}$ \citep{li2010uniform}.
The convergence in other norms such as Frobenius norm can be found in \citep{hall2006properties, hall2006properties2, bunea2014sample}. In general, the consistency result does not really depend on the observational design as long as we can get an estimate of the covariance operator whose sup-norm error vanishes as $n$ and $p$ increase. More discussions about cases where a sample covariance is not feasible can be found in Section 7.

 \noindent\textbf{A2.} There is a positive integer $k$ such that the
eigenvalues of $\PopOperator$ satisfies $\lambda_1 >, \dots, > \lambda_k > \lambda_{k+1}\geq,\dots,\geq 0$, with positive eigen-gap  $\delta \coloneqq\min_{1\leq j\leq k}(\lambda_j-\lambda_{j+1}) > 0$.\\
  \noindent\textbf{A3.} $\PopOperator$ is Lipschitz continuous:
  $$|\PopOperator(s,t)-\PopOperator(s',t')|\le
  L\max(|s-s'|,|t-t'|)\,,~~\forall~s,s',t,t'\,.$$
 \noindent \textbf{A4.} The $\ndim$ leading eigenfunctions of $\PopOperator$ have Lipschitz first derivatives:
  $$
  |\phi_j'(t)-\phi_j'(s)|\le L|t-s|\,,~~\forall~1\le j\le \ndim\,.
  $$

\begin{theorem}[$\ell_2$ consistency]
  \label{thm:main-l2}
  Under assumptions (A1-A4), if  $\SmoothPenalty/p^5 \rightarrow 0$ and
   $\SparsePenalty\rightarrow 0$ as $(n,p)\rightarrow\infty$, then
   for $k$ as defined in A2, we have
   \begin{align*}
     &\sup_{1\le j\le k}\norm{\hat \phi_j(t)-\phi_j(t)}_2\stackrel{P}{\rightarrow} 0\,.
   \end{align*}
\end{theorem}

The proof of \Cref{thm:main-l2} is deferred to the Appendix. Here we outline the proof, highlighting some key technical challenges.

Let $u_{j}$ be the unit vector obtained by discretizing and re-normalizing the eigenfunction $\phi_j(t)$. 
We will prove \Cref{thm:main-l2} by proving
$$\sup_{1\le j\le k}\norm{\hat v_j - u_j}_2\stackrel{P}{\rightarrow} 0.$$

Let $\PopMatrix$
be the discretized covariance operator, with a possible constant shift $a$ on the diagonal. Roughly speaking, $\hat v_j$ and $u_j$ approximate the $j$th eigenvector of $\InputMatrix$ and $\PopMatrix$, respectively. We hope to establish the following inequality based on the standard Davis-Kahan $\sin\Theta$ theorem (\cite{Bhatia97}, Theorem VII.3.1),
\begin{align}\label{eq:simple-bound}
  \norm{\hat v_j \hat v_j^T-u_j u_j^T}_F\le \frac{c}{\Gap p}\norm{\InputMatrix - \PopMatrix}_F\le \frac{c}{\Gap} \norm{\InputMatrix-\PopMatrix}_{\infty,\infty}\,,
\end{align}
provided that $\InputMatrix$ and $\PopMatrix$ have eigengap of order $\Gap p$ (\Cref{lem:popmatrix}), where
$\norm{A}_F=\innerp{A}{A}^{1/2}$ is the Frobenius norm and $\norm{A}_{\infty,\infty}$
is the maximum absolute value of all entries in $A$.

However, to rigorously obtain an approximated version of
\eqref{eq:simple-bound} is non-trivial.  First, $u_j$ is not an
eigenvector of $\PopMatrix$ because of the discretization error. The discretization error will be explicitly tracked in all subsequent analysis
when comparing $u_j$ with $\hat v_j$, for example, in the characterization of  population PCA problem (\Cref{lem:approx-curvature}).
Second, $\hat v_j$ is obtained by solving a penalized eigenvector problem over the deflated Fantope, and hence
is not directly comparable to its ideal theoretical counterpart $u_j$, which may not be in the feasible set
of problem \eqref{opt2}.  To overcome this difficulty, we will consider a modified version of $u_j$ that is feasible for \eqref{opt2} but still possesses similar smoothness as well as proximity
to the true eigenvector of $\PopMatrix$. Furthermore, the sequential estimation procedure \eqref{opt2} involves deflation based
on estimated projection matrix $\EstProjection_{j-1}$, which carries over estimation error from
previous steps. We will use an induction argument to control the sequential error accumulation.

Due to the sequential error accumulation, the convergence rate involves $\rho_1$ and $\rho_2$ in a complex way. To provide insights for our localized estimation procedure, the following theorem shows how the rate of convergence depends on $\rho_2$ when the smoothing parameter $\SmoothPenalty=0$.
The proof is included in the Appendix.
\begin{theorem}[Rate of convergence]
  \label{thm:rate}
  Under assumptions (A1-A4), if  
  $\SmoothPenalty=0$ and $\SparsePenalty\rightarrow 0$ as $(n,p)\rightarrow\infty$, then for $k$ as defined in A2, we have
   \begin{align*}
     &\sup_{1\le j\le k}\norm{\hat \phi_j(t)-\phi_j(t)}_2
     =O_P(e_n+\SparsePenalty+p^{-1})\,.
   \end{align*}
\end{theorem}
The three parts in the rate of convergence correspond to covariance estimation error, bias caused by localization penalty, and discretization error, respectively.
According to the discussion after Assumption A1, the sup-norm covariance estimation error
$e_n$ can be made as small as $\sqrt{\log p/n}$ or $\sqrt{\log n/n}$ depending on
the estimating method used and observation scheme. Thus if $p$ grows at the same or higher
order than $\sqrt{n}$, and $\rho_2=O(e_n)$, our LFPCA estimate achieves an error rate of $e_n$ within a logarithm factor from the standard FPCA error rate.

\section{Numerical Study}\label{sec:simulation}
\begin{figure}[t]
\centerline{
\includegraphics[scale=0.45]{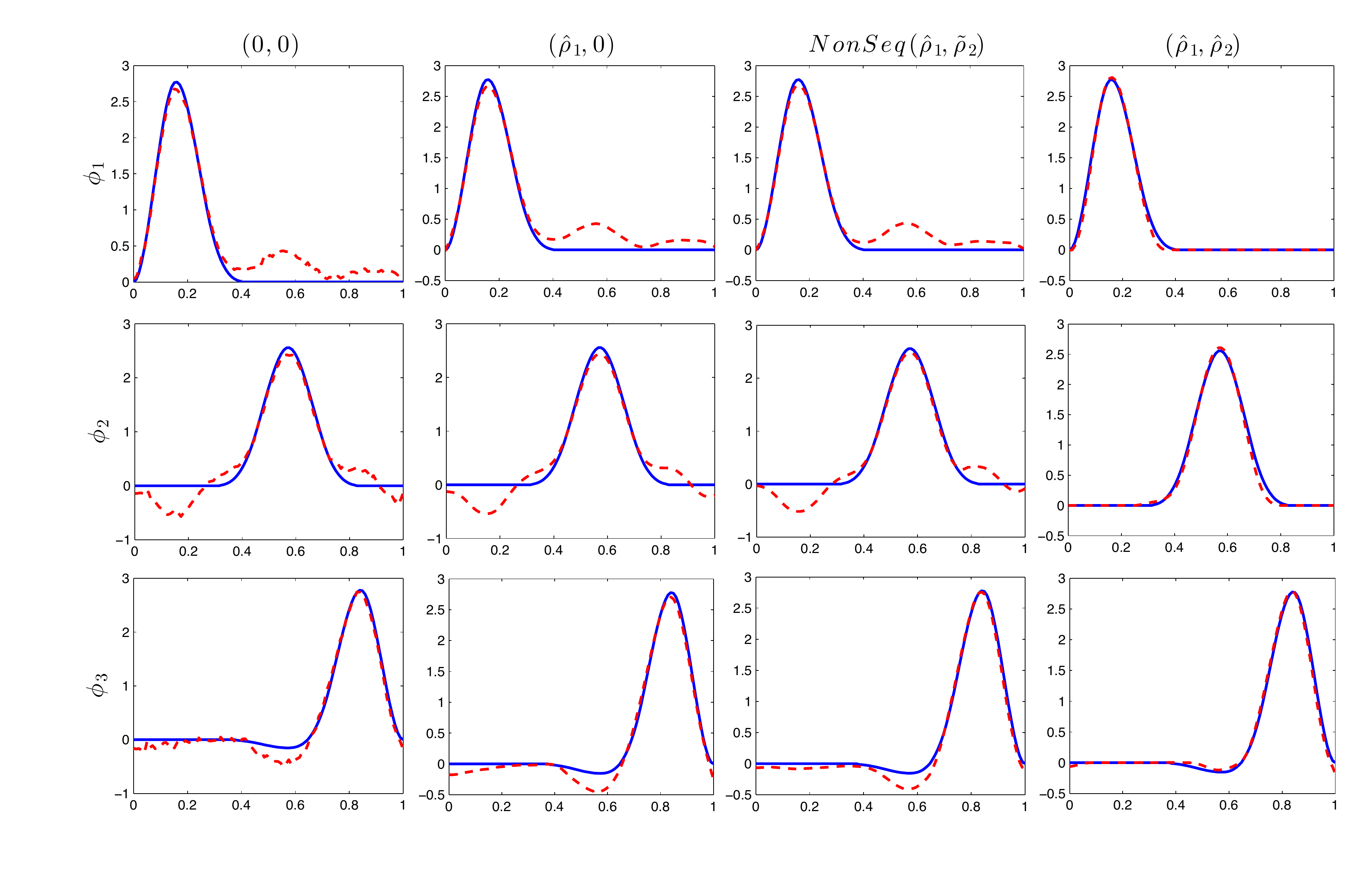}
}
\caption{True (blue-solid) and estimated (red-dashed) eigenfunctions $\phi_j(t)$, $j=1, 2, 3$,  in one run of  {\it Simulation I}, with $n=100$, $p=100$ and $\sigma = 1$, by four different methods as described in Section 5. Tuning parameters are chosen by 5-fold cross-validation.}
\label{Fig:Eigens}
\end{figure}

To illustrate our methods for localized functional principal component analysis, we conduct simulations under two scenarios, {\it Simulation I: localized case} and {\it Simulations II: non-localized case}. For {\it Simulation I},
data $\{Y_{il}, i=1,\dots, n, l = 1, \dots, p \}$ are generated according to model (\ref{eq:obs-model}),
where $t_{l}$ are equally spaced observational points on $[0,1]$. We set $\mu(t) = 0$, $\xi_{ij}\sim N(0,\lambda_j)$, independent, with $\lambda_j$ taken from $(4^2,\ 3^2,\
2.5^2,\ 1.25^2,\ 1, \ 0.75^2,\ 0.5^2,\
0.25^2)$ and $\lambda_j = 0 $ for $j>8$, and the measurement errors
$\epsilon_{il} \stackrel{iid}{\sim} N(0,\sigma^2)$.
We generate the eigenfunctions as follows.  Let 
$\tilde \phi_1(t) =B_3(t)$, $\tilde \phi_2(t) = B_6(t)$ and $\tilde \phi_3(t) = B_9(t)$,  where $B_b(t)$ is the $b$th cubic B-spline basis on $[0,1]$, with 8 equally spaced
interior knots. For $j > 3$, $\tilde \phi_j(t)=\sqrt2\cos((j+1)\pi
t)$ for odd values of $j$ and $\tilde \phi_j(t)=\sqrt2\sin( j \pi t)$ for even values of $j$, $0\leq t\leq 1$.
Then $\phi_j(t),~1\leq j \leq 8,$ are obtained by applying Gram-Schmidt orthonormalization on the set of $\tilde \phi_j(t),~1\leq j\leq 8$. For {\it Simulation II}, we use $\phi_j(t)=\sqrt2\cos((j+1)\pi
t)$ for odd values of $j$ and $\phi_j(t)=\sqrt2\sin( j \pi t)$ for even values of $j$, $0\leq t\leq 1$, for $1\leq j\leq 8$. The rest is the same as {\it simulation I}.

We investigate the performance of the proposed LFPCA under varying combinations of sample size $n$, number of observations per curve $p$, and noise level $\sigma^2$. More importantly, we compare the estimates given by four different methods (i) $(\rho_1, \rho_2) = (0,0)$ corresponds to the ordinary PCA estimation directly obtained from the sample covariance;
(ii) $(\rho_1, \rho_2) = (\hat\rho_1, 0)$ corresponds to the smoothed eigenfunction estimation without localization, where $\hat \rho_1$ was chosen by 5-fold cross-validation as discussed in Section 3.2. Empirically, we found these estimated eigenfunctions almost identical to those estimated from a smoothed covariance function or pre-smoothed individual curves; (iii) {\it NonSeq} corresponds to the subspace method developed in \cite{vu2013fantope}. For comparison purpose, we incorporate the roughening matrix, i.e., use $S-\rho_1 D$ as input matrix in \Cref{eq:pen-pca-conv} with $\hat\rho_1$ chosen by 5-fold cross-validation, the same as used in (ii) and (iv). The sparse tuning parameter  is chosen by 5-fold cross validation $\tilde \rho_2 = \argmax_{\rho \in \mathcal A_2}\sum_{v=1}^{V}\innerp{H^{(-v)}(\rho_1, \rho)}{S^{v}}$. We also note that their proposed method only outputs $k$ basis vectors of the $k$-dimensional subspace, and one needs to rotate that basis to obtain eigenvectors.  (iv) $(\rho_1, \rho_2)  = (\hat \rho_1, \hat \rho_2)$ corresponds to the proposed LFPCA with tuning parameters selected by 5-fold cross-validation as detailed in (\ref{hatrho1}) and (\ref{hatrho2}).
Each setting is repeated 200 times to assess the average performance. The number of included components $k$ is chosen to account for at least $85\%$ of the total variance, i.e. $\sum_{j=1}^{k}FVE(\hat v_j) \ge 85\%$, where FVE is defined in \Cref{FVE}. The selected number of $k$ is quite robust among all simulation settings, and the average number over 200 simulations is 3.01.

\begin{figure}[t]
\centerline{
\includegraphics[scale=0.5]{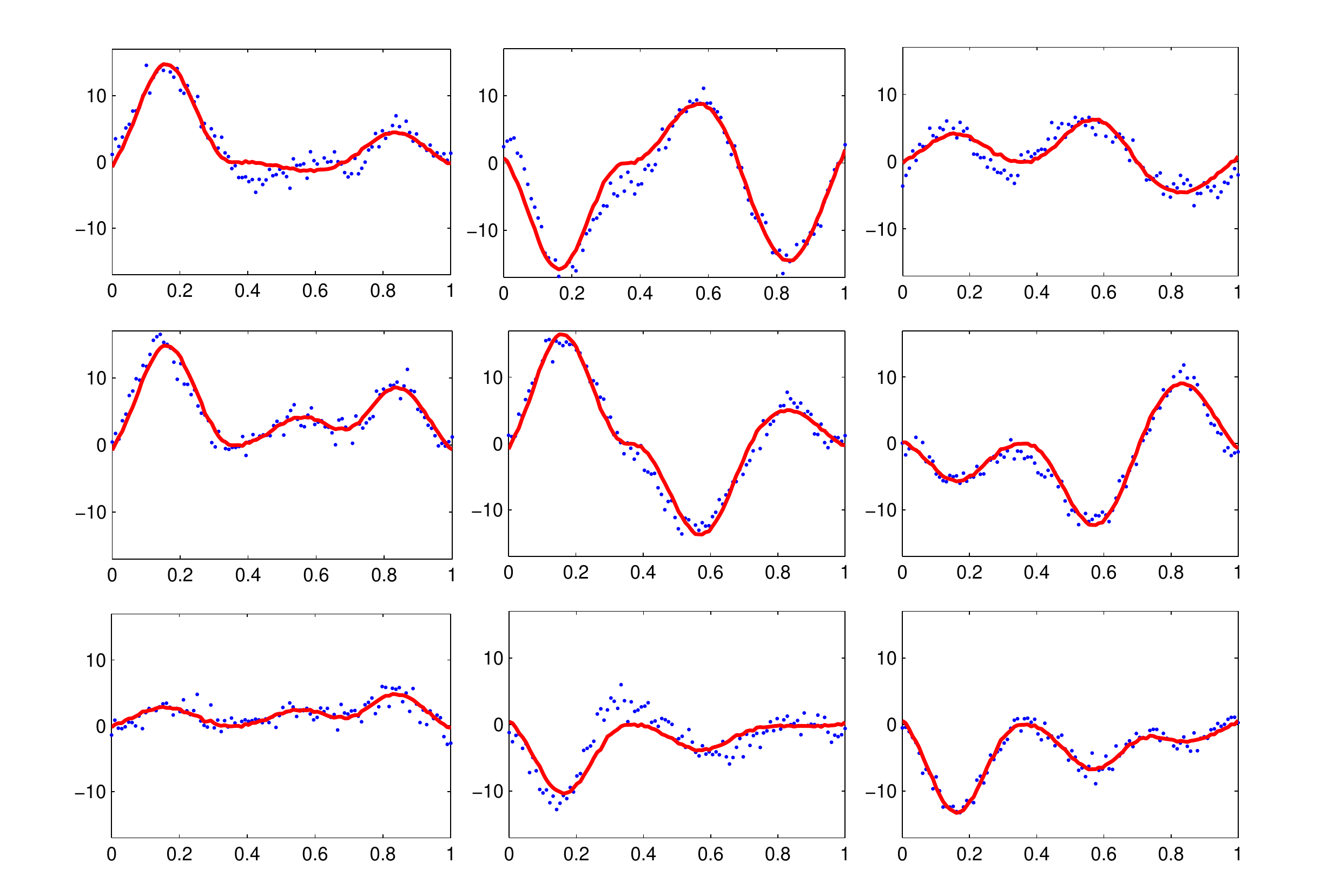}
}
\caption{Noisy observations and recovered functions $\hat X_i(t)$ (red-solid) for nine randomly selected subjects, as obtained in one run of {\it simulation I} with $n=100$, $p=100$ and $\sigma = 1$, using $(\hat \rho_1,\hat \rho_2)$ chosen by 5-fold cross-validation.}
\label{Fig:fittedX}
\end{figure}

\begin{figure}[t]
\centerline{
\includegraphics[scale=0.4]{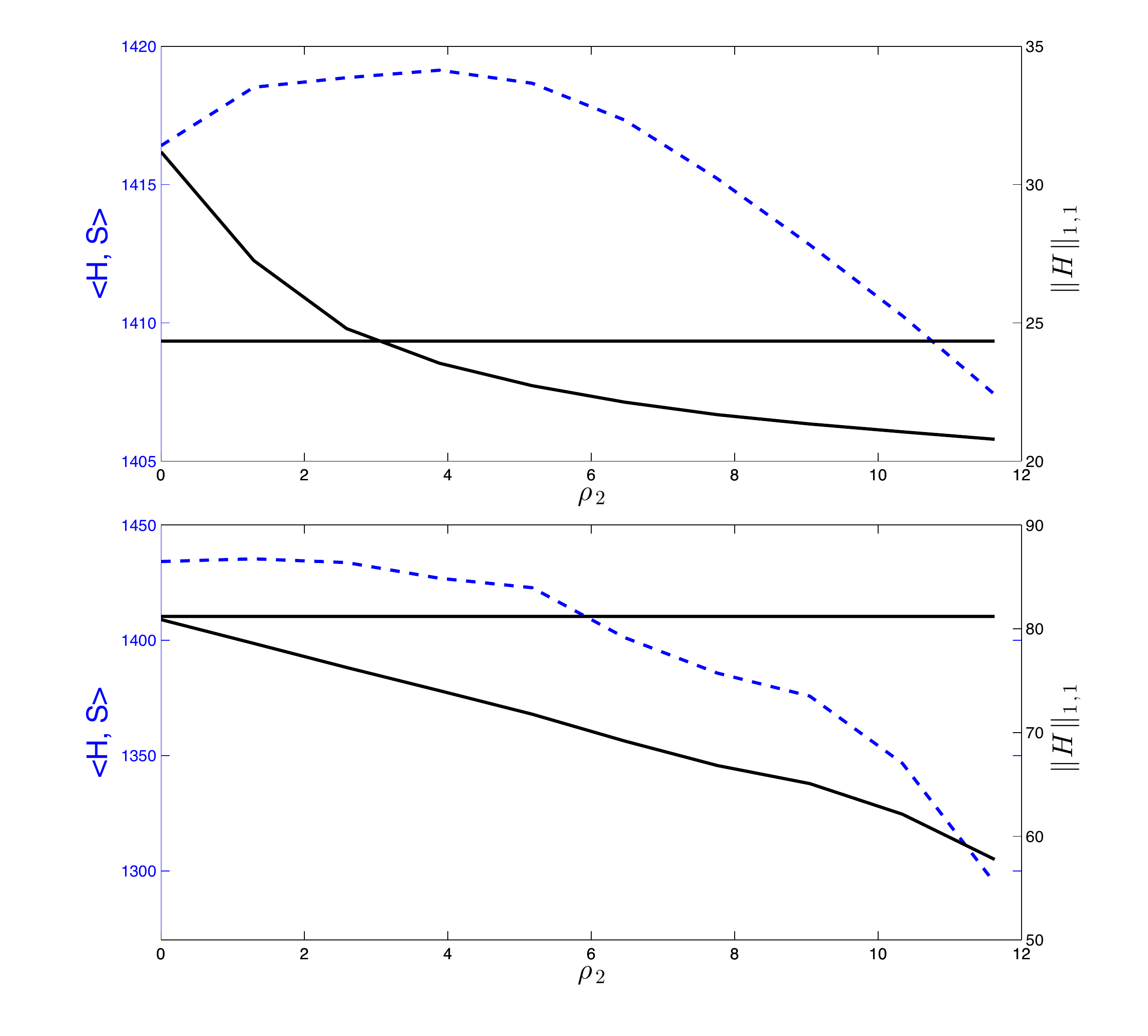}
}
\caption{Performance of the 5-fold cross-validation to choose $\rho_2$ for $\phi_1$. 
The top panel is for {\it Simulation I: localized case}: The peak of the 
cross-validated inner product $\inner{H}{S}$ (dashed blue line, left $y$ label) 
corresponds to $\hat\rho_2 = 3.8$.  The estimated $\phi_1$ is more localized as 
$\rho_2$ increases, and an ideal $\rho_2$ is where the $\norm{\cdot}_{1,1}$ of 
estimated $H$ meets that of true discretized projection matrix corresponding to 
$\phi_1$ (indicated by the horizontal line). The bottom panel is for {\it Simulation 
II: non-localized case}: $\hat\rho_2 = 1.2$. }
\label{Fig:choserho2}
\end{figure}

The estimated eigenfunctions $\phi_j(t),~j=1,2,3,$ from a typical run of {\it Simulation I} with $p = 100$, $n= 100$, and $\sigma = 1$ are visualized in \Cref{Fig:Eigens}. The four columns correspond to results given by the four methods described above. One can clearly see the improvement by adding smoothing and localization penalties. The 5-fold cross-validation choice $(\hat \rho_1, \hat\rho_2)$ leads to almost perfect recovery of the true eigenfunctions. The fitted curves $\hat X_i(t) = \hat \mu(t) + \sum_{j=1}^{k}\hat \xi_{ij}\hat\phi_j(t)$ for nine randomly chosen subjects are shown in Figure \ref{Fig:fittedX}, where $\hat\phi_j(t)$ are obtained using $(\hat \rho_1, \hat\rho_2)$. It demonstrates accurate recovery of the true curves $X_i(t)$.

\begin{table}[t]
\caption{Results for simulation: reporting the median of errors ($\norm{\phi_j-\hat\phi_j}_2$) for $\phi_j$, $j = 1,2,3,$ (with median absolute deviations in parentheses) over 200 simulation runs, with $\sigma =1$, $p= 100$, and varying sample sizes $n$,  where {\it NonSeq} is the subspace method developed in \cite{vu2013fantope}. }\label{simu}
\begin{center}
\begin{tabular}{|c|c|cccccc|}
  \hline
  \hline
  \multicolumn{2}{|c|}{}&\multicolumn{3}{c}{\it Simulation I: Localized }&\multicolumn{3}{c|}{\it Simulation II: Non-localized }\\
\multicolumn{2}{|c|}{}& $n=50$ & $n=100$ & $n=200$ & $n= 50$ & $n=100$ & $n= 200$ \\
 \hline
 & (0, 0)  & 0.22 (0.09) & 0.18 (0.06) & 0.13 (0.04) & 0.22 (0.08) & 0.15 (0.05) & 0.12 (0.04) \\ 
  $\phi_1$ &  ($\hat\rho_1$, 0) & 0.22 (0.09) & 0.18 (0.07) & 0.13 (0.04) & 0.21 (0.08) & 0.15 (0.06) & 0.11 (0.04) \\ 
  & {\it NonSeq}& 0.22 (0.09) & 0.18 (0.07) & 0.13 (0.04) & 0.21 (0.08) & 0.15 (0.06) & 0.11 (0.03) \\ 
   &($\hat\rho_1$, $\hat\rho_2$)& 0.12 (0.05) & 0.10 (0.03) & 0.06 (0.02) & 0.22 (0.08) & 0.15 (0.05) & 0.13 (0.03) \\
   \hline 
  & (0, 0) & 0.42 (0.16) & 0.30 (0.13) & 0.20 (0.07) & 0.42 (0.16) & 0.29 (0.12) & 0.19 (0.07) \\ 
 $\phi_2$ &  ($\hat\rho_1$, 0) & 0.41 (0.16) & 0.30 (0.13) & 0.20 (0.07) & 0.41 (0.15) & 0.28 (0.11) & 0.19 (0.06) \\ 
  & {\it NonSeq} & 0.40 (0.16) & 0.30 (0.12) & 0.20 (0.07) & 0.41 (0.15) & 0.28 (0.11) & 0.19 (0.07) \\ 
   &($\hat\rho_1$, $\hat\rho_2$) & 0.26 (0.17) & 0.14 (0.07) & 0.11 (0.05) & 0.41 (0.17) & 0.28 (0.11) & 0.20 (0.07) \\ 
   \hline
   & (0, 0)  & 0.37 (0.12) & 0.27 (0.11) & 0.18 (0.07) & 0.31 (0.10) & 0.26 (0.09) & 0.18 (0.06) \\ 
 $\phi_3$ &  ($\hat\rho_1$, 0) & 0.36 (0.12) & 0.27 (0.11) & 0.18 (0.08) & 0.30 (0.10) & 0.26 (0.09) & 0.17 (0.06) \\ 
 & {\it NonSeq} & 0.33 (0.14) & 0.27 (0.11) & 0.18 (0.07) & 0.30 (0.10) & 0.26 (0.09) & 0.18 (0.06) \\ 
   &($\hat\rho_1$, $\hat\rho_2$) & 0.24 (0.15) & 0.14 (0.08) & 0.09 (0.04) & 0.31 (0.11) & 0.26 (0.09) & 0.18 (0.06) \\ 
   \hline
   \hline
\end{tabular}
\end{center}

\end{table}

 To better quantify the performance of estimating $\phi_j(t)$ we report the $\ell_2$ distance $\norm{\phi_j(t) - \hat\phi_j(t)}_2$.
The medians of the errors over 200 simulation runs are reported in Table \ref{simu}.
The results are quite similar for different levels of $p$ and $\sigma$, so only results for $p=100$ and $\sigma = 1$  are reported with various sample size $n$.  The errors are found to decline with increasing
sample size $n$, as expected. 
 For {\it Simulation I: localized case}, the proposed LFPCA with $\hat \rho_2$ chosen by cross-validation significantly outperforms other methods. The $\hat \rho_2$ chosen by cross-validation well approximates the true localization level of the eigenfunctions; see Figure \ref{Fig:choserho2} for an illustration and detailed discussions can be found in section 3.2. The {\it NonSeq}, a subspace method proposed by \cite{vu2013fantope}, does not perform well since the union of the subdomains under consideration is the entire domain, not `sparse' at all in their setting. The results demonstrate the advantage of our proposed sequential method. For {\it Simulation II: non-localized case}, 5-fold cross-validation method combined with the proposed LFPCA choose $\hat\rho_{2j}$ to be very close to 0, and as expected, the $\ell_2$ errors of the four methods are almost identical.


\section{Data Applications}\label{sec:data}
\subsection{Application to Country Mortality Data}

\begin{figure}[t]
\centerline{
\includegraphics[scale=0.5]{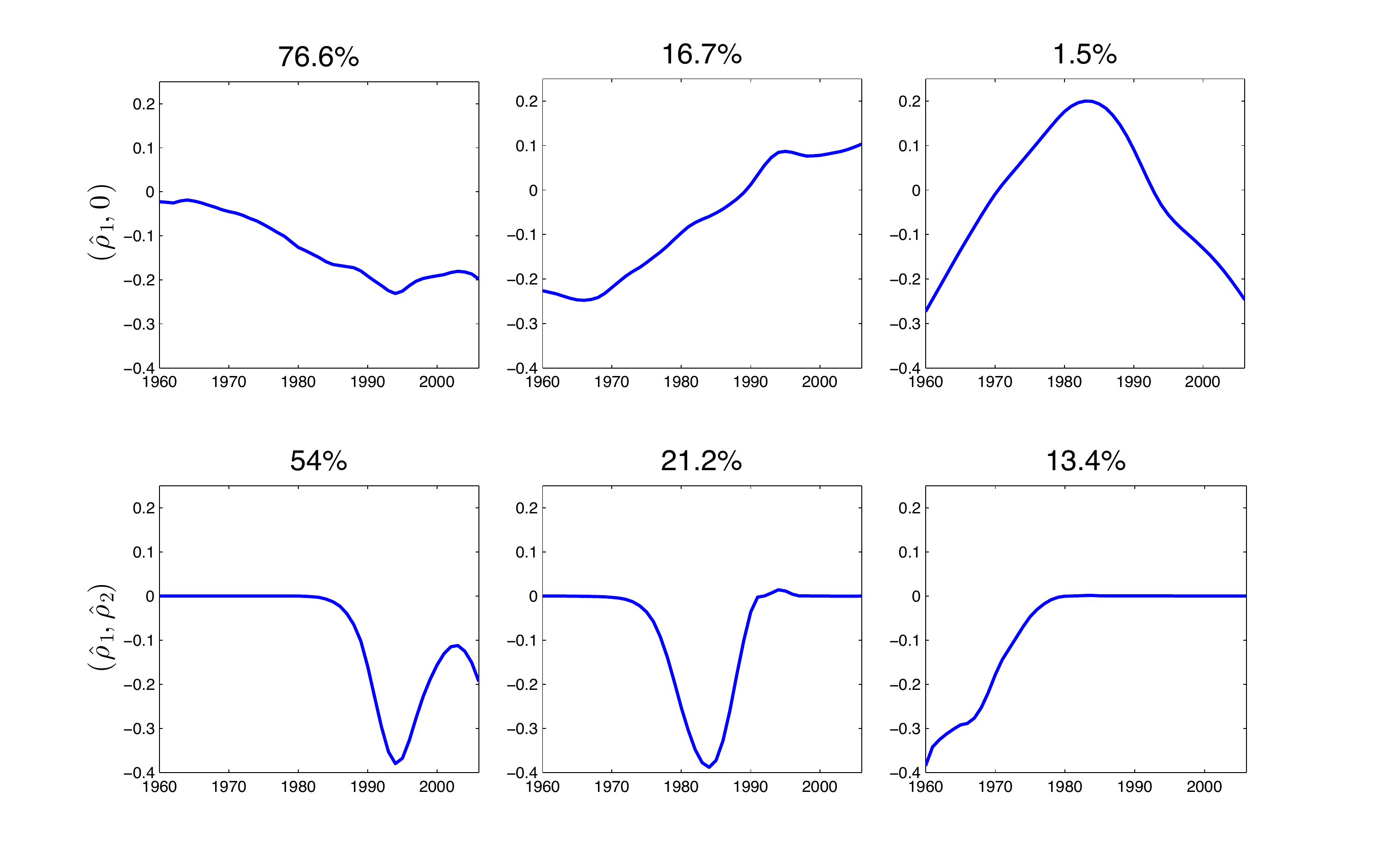}
}
\caption{Top Row: Estimated eigenfunctions for the mortality data, $\hat\rho_1$ is chosen by 5-fold cross validation; Bottom Row: Estimated orthogonal basis functions, $\hat\rho_2$ is chosen to maintain rFVE  at 70\% in (\ref{rFVE}), and the number of components $k = 3$ is chosen to explain at least 85\% of the total variance.}
\label{Fig:mortality}
\end{figure}
The analysis of human mortality is important in assessing the future demographic
prospects of societies, and quantifying differences between countries with regard to
the overall public health measure. Functional data analysis approaches have been previously applied to
study mortality data \citep{hyndman2007robust,chiou2009modeling,chen2012modeling}.
To study the variational modes of mortality rates across countries over years, we applied the proposed LFPCA method to
period life tables for 27 countries, with rates of mortality at age 60 available for each of the
calendar years from 1960 to 2006. The data were obtained from the Human
Mortality Database (downloaded on
March 1, 2011), maintained by University of California,
Berkeley (USA), and Max Planck Institute for Demographic Research (Germany).
The data is available at \url{www.mortality.org} or \url{www.humanmortality.de}, with detailed description in \cite{wilmoth2007methods}.

Let $X_i(t)$ denote the
mortality rate in the $i$th country for subjects at age 60 during calendar
year $t$, where $1960\leq t\leq 2006$. We directly compute the sample covariance
matrix $S$ from the observed data and apply the proposed algorithm to solve problem (\ref{opt2}).
The $\hat \rho_1$ chosen by 5-fold cross-validation (\ref{hatrho1}) is always used to ensure a relatively smooth estimate of the eigenfunction, and the solution path along different levels of localization is investigated. The 5-fold cross-validation method gave $\hat \rho_2 = 0$, indicating the eigenfunctions are not exactly localized. The estimated eigenfunctions without localization penalty are given in the top row of Figure \ref{Fig:mortality}. We then choose $\hat\rho_2$ by the second method as defined in \Cref{hatrho2_2} with $a = 30\%$. The estimated localized basis functions, as visualized in the bottom row of Figure \ref{Fig:mortality}, reveal several
interesting features. The first localized basis function
$\phi_1(t)$, explaining 54\% of the total variance, indicates that a big variation
of the mortality functions $X_i(t)$ around their mean function happens around mid 1990s. The second basis function
$\phi_2(t)$ with a mode around 1980s accounts for 21.2\% of the total variation. The third
basis function $\phi_3(t)$ characterizes variation of mortality around 1960s. Although the first localized basis function explains less variance than the first leading eigenfunction, only retaining 70\% of the capability in return of localization, the lost proportion is picked up by the second component. The second component could have explained 30.3\% of the variance without localization. Therefore, we only need three localized eigenfunctions to account for more than 85\% of the total variance. 

\subsection{Application to Berkeley Growth Data}

The smooth nature of growth curves has been explored in various previous statistical
analyses, including functional data analysis approaches. 
We apply the proposed LFPCA method to the Berkeley growth data \citep{tuddenham1954physical}.
These data contain height measurements for 54 girls, with 31
measurements taken between ages 1 year and 18 years.  A sample covariance matrix is computed based on equally spaced measurements at every half year from interpolated curves and then the proposed algorithm is used to solve problem (\ref{opt2}).
The solution path along different levels of localization is investigated.
The estimated eigenfunctions without localization penalty are visualized in the top row of Figure \ref{Fig:growth}, and the estimated localized eigenfunctions are given in the bottom row of Figure \ref{Fig:growth}. 
The localization level is chosen to maintain $rFVE = 70\%$ in \eqref{rFVE}. The total number of components $k = 2$ is chosen to explain total variation of 85\%. 
The first estimated localized basis function, explaining 70.1\% of the variation, indicates a variational mode in girls' growth around age twelve, which obviously matches the well known pubertal growth spurt. The second estimated localized basis function, explaining 18.1\% of the variation, is localized around ages five and six, which remarkably matches the mid-growth spurt previously studied by many researchers \citep{gasser1985analysis,sheehy1999analysis}. The mid-growth spurt is a growth phenomenon during early childhood, expressed by a mild transitory acceleration of growth velocity between years five and eight.  The individual variations in timings, durations and intensities of mid-growth spurt are of great interest and some hypotheses have been proposed for the explanation of individual differences \citep{muhl1991mid}.
This particular ``mode of variation" is not obvious in standard FPCA. The proposed LFPCA method finds a balance between interpretability (localization) and amounts of variance explained.

\begin{figure}[t]
\centerline{
\includegraphics[scale=0.4]{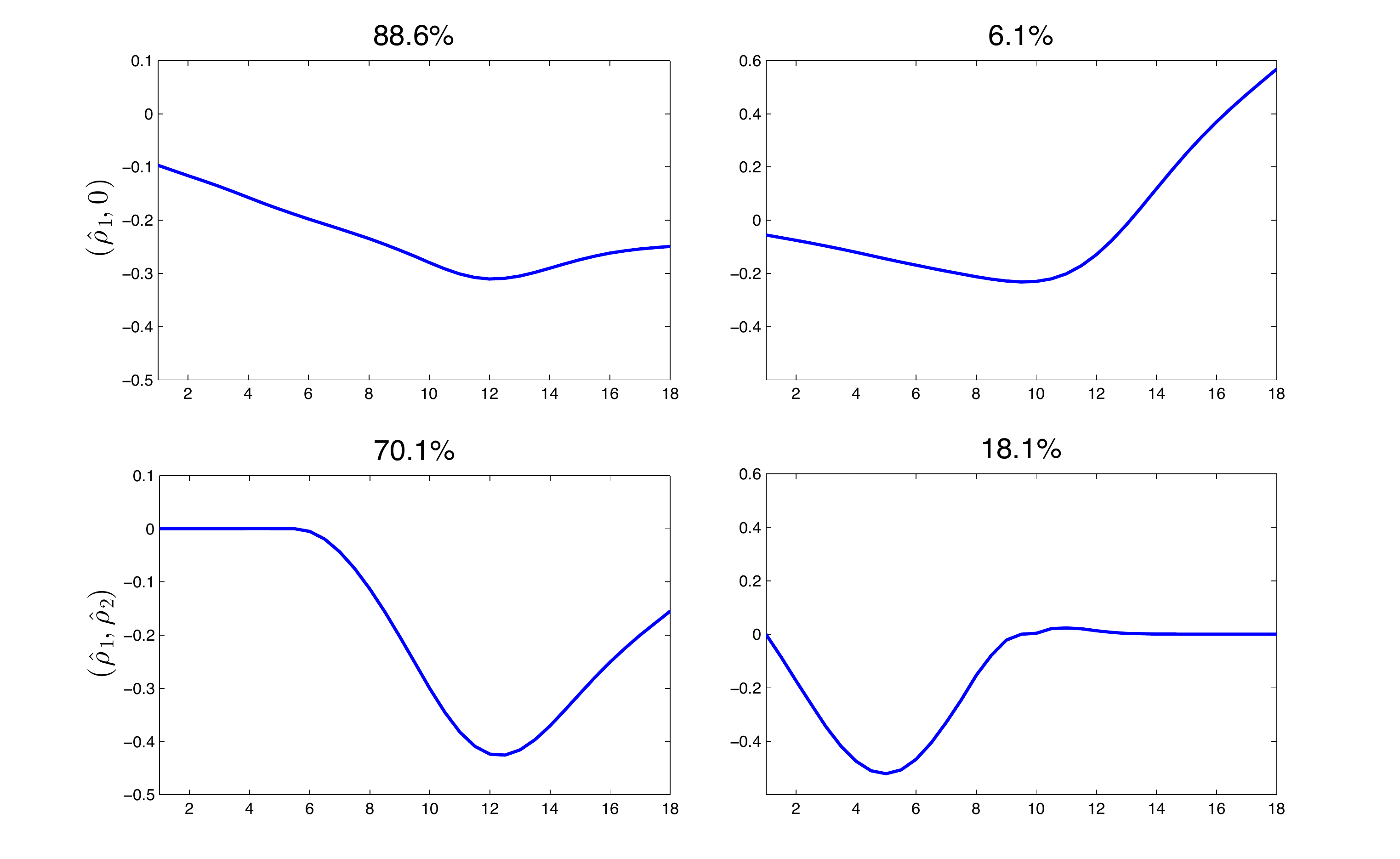}
}
\caption{Top Row: Estimated eigenfunctions for the growth data, $\hat\rho_1$ is chosen by 5-fold cross validation; Bottom Row: Estimated orthogonal basis functions, $\hat\rho_2$ is chosen to maintain rFVE at 70\%, and the number of components $k = 2$ is chosen to explain at least 85\% of the total variance.}
\label{Fig:growth}
\end{figure}

\section{Discussion}\label{sec:discussion}
In this paper, we propose a localized functional principal component analysis through a Deflated Fantope Localization method, where sequentially obtained eigenfunctions have guaranteed orthogonality and are allowed to be supported on different localized subdomains. As mentioned in \Cref{sec:formulation}, the deflated Fantope $\D{\Pi}$ can easily be generalized to a $d$-dimensional version $\mathcal{D}_{\Pi}^{d}$. In some applications, one might be interested to  estimate mutually orthogonal principal subspaces with dimensions $d_1, \dots, d_k$, and each principal subspace is only supported on a subdomain $\mathcal{T}_j\subset \mathcal T$.

Throughout the paper, we mainly focus on a dense and regular functional data design where $p$ equally spaced observations are recorded on each curve. Most commonly seen functional data have this design and a sample covariance can be easily computed from the discrete and possibly noisy observations. Our proposed formula (\ref{opt2}) takes the sample covariance $S$ and outputs smooth and localized estimates of eigenfunctions.
In fact the proposed method puts rather minimal requirement on the input matrix $S$ and does not rely on the design of observations. Consider the discretized version of $\PopOperator$ by evaluating on $p\times p$ equally spaced grid points. The estimation error of the eigenfunctions directly ties to the maximum entry-wise error of the input matrix $S$.
Here we briefly discuss several scenarios where a sample covariance is not feasible. (i) For dense but irregularly observed functional data, one can simply smooth each curve \citep{ramsay2005functional} or interpolate between points to get $p$ equally spaced observations, and then a $p\times p$ covariance estimate $S$ can be computed. This is what we have done for the Berkeley growth data. (ii) For sparse functional data where the observations are recorded at random and sparse time points, individual smoothing or interpolation is impossible. But a uniformly consistent covariance estimation is possible by, for example, two-dimensional smoothing methods \citep{yao2005functional,li2010uniform}. Our proposed LFPCA can then be applied by taking $S$ as the discretized version of the smooth covariance estimator. (iii) For ultra dense and noisy data, the independent measurement errors accumulate if one uses sample covariance computed from the raw measurements.  Moreover, using a large $p\times p$ matrix is not computationally efficient. We recommend  performing pre-smoothing or pre-binning on individual curves and choosing a grid with a moderate size $p$.

We proposed two methods of choosing the localization parameter $\rho_2$. When the cross-validation method chooses $\rho_2 =0$, it roughly means that the true eigenfunctions are not localized. Then for a fixed number of $a$, we find a set of orthogonal basis functions that retain the ability to explain a fair amount of variance (at least a $(1-a)$ proportion) and are localized. In this case, the outcome would depend on the choice of $a$ and they should not be interpreted as the true eigenfunctions. Rather, these localized basis functions and the corresponding projection scores have ready interpretation with domain knowledge.

\section{Appendix}

\begin{proof}[\bf Proof of \Cref{lem:char-dfantope}]
Because $\D{\Pi}$ is a compact set, we know that
 $B=\Proj_{\D{\Pi}}(A)$ exists and is unique.
 Let $G = U^T B U$,  then we have $G\in \Fantope{1}$ and $B=U G U^T$.
Note that $G$ minimizes $\norm{A-UGU^T}_F^2$ over $\Fantope{1}$ and
  \begin{align*}
    \norm{A-UGU^T}_F^2 = & \norm{A}_F^2 - 2\innerp{A}{UGU^T} + \norm{UGU^T}_F^2\\
     = & \norm{A}_F^2 - 2\innerp{U^TA U}{G} + \norm{G}_F^2\\
     = & \norm{A}_F^2 - \norm{U^T A U}_F^2 +\norm{U^T A U-G}_F^2\,.
  \end{align*}
  Therefore, $G$ is the projection of $U^T A U$ onto $\Fantope{1}$ and
  by Lemma 4.1 of \cite{vu2013fantope} we have
  $$
  G = \sum_{i=1}^{p-d} \gamma_i^+(\theta) \eta_i \eta_i^T
  $$
  with $\gamma_i$, $\eta_i$, and $\theta$ specified in the theorem.
\end{proof}

\begin{proof}[\bf Proof of \Cref{thm:convergence}]
For any given $\step>0$, define the augmented Lagrangian of \eqref{eq:H_i-equiv} as
$$
L_{\step}(\PrimalVariable, \DualVariable, Y) =\ConvexIndicator{\D{\Pi}}(\PrimalVariable)-\innerp{S}{\PrimalVariable} + \SparsePenalty \norm{\DualVariable}_{1,1} +\innerp{Y}{\PrimalVariable - \DualVariable}
+\frac{\step}{2}\norm{\PrimalVariable-\DualVariable}_F^2\,.
$$
The update steps in \Cref{alg:admm} now reads, letting $W^{(r)}= \step Y^{(r)}$,
\begin{align*}
  \PrimalVariable^{(r)} & = \arg\min_{\PrimalVariable} L_\step(\PrimalVariable,\DualVariable^{(r-1)},W^{(r-1)})\,,\\
  \DualVariable^{(r)} & = \arg\min_{\DualVariable} L_\step(\PrimalVariable^{(r)},
  \DualVariable, W^{(r)})\,,\\
  Y^{(r)} & = Y^{(r-1)} + \step(\PrimalVariable-\DualVariable)\,.
\end{align*}
It is obvious that $ \ConvexIndicator{\D{\Pi}}(\PrimalVariable)-\innerp{S}{\PrimalVariable}$
and $\SparsePenalty \norm{\DualVariable}_{1,1}$ are closed, proper, and convex functions.
Here we say a function $f$ is closed, proper and convex if $\{(x,t):f(x)\le t\}$ is a closed non-empty convex set (\cite{Boyd-ADMM}, Section 3.2).

By strong duality, we can find a primal-dual pair of $L_0(\PrimalVariable, \DualVariable, Y)$, denoted as $(\PrimalVariable^{**},\DualVariable^{**}, Y^{**})$.
It then follows from the primal and dual optimality that
$(\PrimalVariable^{**},\DualVariable^{**}, Y^{**})$ is a saddle point of $L_0$ and hence
by Section 3.2.1 of \cite{Boyd-ADMM}, we have
$$
\DualVariable^{(r)}\rightarrow \DualVariable^{*}~~\text{and}~~ \PrimalVariable^{(r)}-\DualVariable^{(r)}\rightarrow 0,~~\text{as}~~t\rightarrow\infty\,,
$$
where $(\PrimalVariable^{*},\DualVariable^{*})$ is an optimal primal
variable for $L_0$.
\end{proof}

\noindent{\bf Proof of consistency result.}
To prove \Cref{thm:main-l2}, we need some additional lemmas and notation as follows. The proof of lemmas are given after the proof of \Cref{thm:main-l2}.

Let $I_{j}=((j-1)/p,j/p]$ for $j=2,...,p$ and $I_{1}=[0,1/p]$. We define $\phi_{j}^*(t)=\phi_j(t_{i})$ for $t\in I_i$, $u_{j}^*=p^{-1/2}(\phi_j(t_{1}),\phi_j(t_{2}),...,\phi_j(t_{p}))^T$,
and $u_j = u_j^*/\norm{u_j^*}_2$.  Let $\tilde\PopOperator : [0,1]^2\mapsto [0,\infty)$ be such that
$\tilde\PopOperator(s,t)= \PopOperator(t_{i},t_{j}), \text{if}~s\in I_i,t\in I_j\,.$
Define the discretized and diagonal-shifted covariance matrix $\PopMatrix$ by $\PopMatrix(l,l')=\PopOperator(t_{l},t_{l'})+a\mathbf 1(l=l')\,.$
Let $\tilde\phi_{j}$ be eigenfunctions of $\tilde\PopOperator$ and $v_{j}$ be eigenvectors of $\PopMatrix$.
Then $p^{-1/2}\tilde\phi_{j}(t)$ is the $i$th entry of $v_{j}$ if $t\in I_i$. If we further denote the $j$th eigenvalue of $\tilde\PopOperator$
by $\tilde\lambda_{j}$, then $(p\tilde\lambda_{j}+a,v_{j})$ is an eigenvalue-eigenvector pair
of $\PopMatrix$.

Let $\PopProjection_j=\sum_{i=1}^j v_i v_i^T$, and $\PopMatrix_j=\PopMatrix - \PopProjection_{j-1} \PopMatrix \PopProjection_{j-1}$.  Define $\PopProjection_0=0$ for convenience. 
For any measurable $B:[0,1]^2\mapsto \Real$, let $\norm{B}_{\rm HS}
=[\int_{[0,1]^2} B(s,t)^2 ds dt]^{1/2}$ be the Hilbert-Schmidt norm.

\begin{lemma}\label{lem:u-v}
  Under Assumptions (A1-A4), let $c_0=L(2/\Gap+1)$ we have
  for $p$ large enough, we have
  $\norm{v_{j}-u_{j}}_2\le c_0 p^{-1}\,,~~\text{for}~1\le j\le \ndim$.
\end{lemma}

\begin{lemma}\label{lem:popmatrix}
Under Assumptions (A1-A3), when $p$ is large enough we have $\norm{\PopMatrix}_{F}^2 \le c_2^2 p^2$, where $c_2^2=a^2+4\norm{\PopOperator}_{\rm HS}^2+L^2$.
Moreover, the gap between the $j$th and $(j+1)$th eigenvalues of
  $\PopMatrix$ is at least $p \Gap/2$ for all $1\le j\le \ndim$.
\end{lemma}

The following lemma is elementary and can be found in  \cite{vu2013minimax}.
\begin{lemma}\label{lem:simple}
  Let $u$ and $v$ be vectors of same length with unit norm.  Then
  \begin{align*}
    \frac{1}{\sqrt{2}}\norm{u-v}_2\le \norm{uu^T - vv^T}_F\le \sqrt{2}\norm{u-v}_2\,.
  \end{align*}
  The same holds when $u$, $v$ are functions and $\norm{\cdot}_F$ is replaced by $\norm{\cdot}_{\rm HS}$.
\end{lemma}

The next lemma characterizes $u_i$ as an approximate leading eigenvector of $\PopMatrix_i$.  It extends Lemma 4.2
of \cite{vu2013minimax}.
\begin{lemma}[Approximate curvature]\label{lem:approx-curvature}
  Let $H_j$ be a solution to \eqref{eq:H_i}.
  Then under Assumptions (A2-A4), for $p$ large enough,
  \begin{align*}
  \frac{p\Gap}{8}\norm{\PrimalVariable_j - u_j u_j^T}_F^2 -\frac{3c_0 c_2}{2} \le & \innerp{-\PopMatrix_j}{\PrimalVariable_j- u_j u_j^T}\,,~~\forall~~1\le j\le \ndim\,.\end{align*}
\end{lemma}

\begin{proof}[\bf Proof of \Cref{thm:main-l2}]
  The claim follows if we show that $\sup_{1\le j\le \ndim}\norm{\hat v_j - u_j}_2=o_p(1)$.

For simplicity denote $e_n\coloneqq\norm{\InputMatrix-\PopMatrix}_{\infty,\infty}$, which is $o_p(1)$
by assumption (A1).
Let $\PopProjection_0=\EstProjection_0=0$, $\hat v_0=0$, $\beta_0=\epsilon_0=0$.
We use induction to show that
 there exist $(\epsilon_i,\beta_i:0\le i\le \ndim)$ such that
\begin{align}
  &\sup_{0\le i\le \ndim} \epsilon_i=o_p(1),~~\sup_{0\le i\le \ndim}\beta_i=o_p(p)\,,\\
  &\max\{\norm{\hat v_i - v_i}_2,\norm{\EstProjection_i-\PopProjection_i}_2\}\le \epsilon_i
  \label{eq:induct-l2}\\
  &\SmoothPenalty \innerp{\DiffOperator}{\hat v_i \hat v_i^T}\le \beta_i\,.
  \label{eq:induct-smooth}
\end{align}

Obviously the claim holds for  $\ndim=0$.  Now assume that
the claim holds for $\ndim=j-1$, and $j\ge 1$.  We will construct
$\epsilon_j=o_p(1)$ and $\beta_j=o_p(p)$ satisfying \Cref{eq:induct-l2,eq:induct-smooth}.

Let  $\PopMatrix_j = \PopMatrix - \PopProjection_{j-1}\PopMatrix\PopProjection_{j-1}$ for $j=1,...,k$.
By \Cref{lem:approx-curvature} we have, for $p$ large enough,
\begin{align}\label{eq:3-1}
\frac{p\Gap}{8}\norm{\PrimalVariable_j - u_j u_j^T}_F^2 -\frac{3c_0c_2}{2} \le & \innerp{-\PopMatrix_j}{\PrimalVariable_j- u_j u_j^T}\,.\end{align}

Next we need to control
$\innerp{\PopMatrix-\PopMatrix_j}{\PrimalVariable_j-u_j u_j^T}$
and $\innerp{\InputMatrix}{\PrimalVariable_j-u_j u_j^T}$,
where the first one is small because $H_j$ and $u_j u_j^T$ are nearly
orthogonal to $\PopMatrix-\PopMatrix_j$, and the second term
needs to be controlled by the fact that $H_j$ is a maximizer of
\eqref{eq:H_i}.

For the first term $\innerp{\PopMatrix-\PopMatrix_j}{\PrimalVariable_j-u_j u_j^T}$, by the orthogonality constraint, we have
\begin{align}
  \innerp{\PopMatrix-\PopMatrix_j}{\PrimalVariable_j}
  \le & \lambda_1 \innerp{\PopProjection_{j-1}}{\PrimalVariable_j}
  = \lambda_1|\innerp{\PopProjection_{j-1}-\EstProjection_{j-1}}{\PrimalVariable_j}|
  \le \lambda_{1}\epsilon_{j-1}\le c_2 p\epsilon_{j-1}\,.\nonumber
\end{align}
and similarly
\begin{align}
  \innerp{\PopMatrix-\PopMatrix_j}{u_j u_j^T}=&\innerp{\PopMatrix-\PopMatrix_j}{u_j u_j^T-v_j v_j^T}
  \le  \norm{\PopMatrix_{j-1}}_F\norm{u_j u_j^T - v_j v_j^T}_F\le \sqrt{2}c_2 c_0\,,\nonumber
\end{align}
where the last inequality follows from \Cref{lem:u-v} and \Cref{lem:simple},
and therefore
\begin{align}
  \left|\innerp{\PopMatrix-\PopMatrix_j}{\PrimalVariable_j-u_j u_j^T}\right|\le c_2 p(\epsilon_{j-1}+\sqrt{2}c_0p^{-1})\label{eq:3-3}\,.
\end{align}

Now we turn to the term $\innerp{\InputMatrix}{\PrimalVariable_j-u_j u_j^T}$.
 If we can show that
\begin{align}
 0 \le & \innerp{\InputMatrix}{\PrimalVariable_j-u_j u_j^T}
    - \SmoothPenalty\innerp{\DiffOperator}{\PrimalVariable_j}+
  R_j
  \,,\label{eq:alpha=0-optimality-1}
\end{align}
for some $R_j=o_p(p)$
then we have, combining \cref{eq:3-1,eq:3-3,eq:alpha=0-optimality-1},
\begin{align}
\frac{p\Gap}{8}\norm{H_j-u_ju_j^2}_F^2\le &
\innerp{\InputMatrix-\PopMatrix}{H_j-u_ju_j^T} -\SmoothPenalty\innerp{\DiffOperator}{H_j}
+ R_j'\nonumber\\
\le &
e_n p \norm{H_j - u_j u_j^T}_F-\SmoothPenalty\innerp{\DiffOperator}{H_j}
 + R_j'\,,\label{eq:alpha=0-quadratic}
\end{align}
where
$
R_j'=c_2 p(\epsilon_{j-1}+\sqrt{2}c_0p^{-1})+R_j+3c_0c_2/2\,.$
It follows that
\begin{align}
  \norm{H_j-u_j u_j^T}_F\le \frac{8 e_n }{\Gap}+\sqrt{\frac{8R_j'}{\Gap p}}\,.
\end{align}
Since $\hat v_j \hat v_j^T$ is the closest rank one, unit norm matrix to
$H_j$, we have
\begin{align*}
  &\norm{\hat v_j \hat v_j^T-v_j v_j^T}_F \le
  \norm{\hat v_j \hat v_j^T - u_j u_j^T}_F + \norm{u_j u_j^T - v_j v_j^T}_F\\
  \le & 2\norm{H_j - u_j u_j^T}_F+\sqrt{2}c_0p^{-1}
  \le
  \frac{16 e_n}{\Gap}+2\sqrt{\frac{8R_j'}{\Gap p}}+\sqrt{2}c_0p^{-1}\,,
\end{align*}
and
\begin{align*}
&\norm{\EstProjection_j-\PopProjection_j}_F\le \norm{\EstProjection_{j-1}-\PopProjection_{j-1}}_F+  \norm{\hat v_j \hat v_j^T-v_j v_j^T}_F\\
\le &\epsilon_{j-1}+\frac{16 e_n }{\Gap}+2\sqrt{\frac{8R_j'}{\Gap p}}+\sqrt{2}c_0p^{-1} \eqqcolon \epsilon_{j}\,.
\end{align*}

Now it remains to find $\beta_j$.  Using  \eqref{eq:alpha=0-quadratic} we have
\begin{align*}
  \SmoothPenalty\innerp{\DiffOperator}{H_j}\le &
  e_n p \norm{H_j-u_j u_j^T}_F+R_j'
  \le  2e_n p \epsilon_j + R_j'\,.
\end{align*}
On the other hand, let $\lambda_{j,1}$ be the largest eigenvalue of $H_j$.
Then \begin{align*}
  \frac{\epsilon_j^2}{4}\ge\norm{H_j-u_j u_j^T}_F^2= \norm{H_j}_F^2 -2\innerp{H_j}{u_j u_j^T}+1
  \ge \lambda_{j,1}^2 - 2\lambda_{j,1}+1\,,
\end{align*}
where we use the fact that $\norm{H_j}_F^2 \ge \lambda_{j,1}^2$, and $\innerp{H_j}{u_j u_j^T}
\le \lambda_{j,1}\norm{u_j}_2^2$ (von Neumann trace inequality).  It then follows that
  $\lambda_{j,1}\ge 1-\epsilon_j/2$,
which implies that
\begin{align*}
   \SmoothPenalty\innerp{\DiffOperator}{\hat v_j \hat v_j^T}\le & (1-\epsilon_j/2)^{-1}\SmoothPenalty\innerp{\DiffOperator}{H_j}\le
   (1-\epsilon_{j}/2)^{-1}
   \left(2 e_n p \epsilon_j + R_j'\right)\eqqcolon \beta_j\,.
\end{align*}
Direct verification shows that if $\max_{0\le i\le j-1}\epsilon_i =o_p(1)$, $\max_{0\le i\le j-1}\beta_j=o_p(p)$, and $R_j=o_p(p)$,
 then
$\epsilon_j=o_p(1)$ and $\beta_j=o_p(p)$.

The rest of the proof is to show \eqref{eq:alpha=0-optimality-1}
for some $R_j=o_p(p)$.
The main challenge is that $u_j$ is not in the feasible set of problem \eqref{eq:H_i} and hence $u_j u_j^T$ is not directly comparable to $H_j$
using optimality condition of \eqref{eq:H_i}.
To overcome this difficulty, we consider $\tilde u_j$, a modified version of
$u_j$ so that (a) $\tilde u_j$ is close to $u_j$ in $\ell_2$ norm;
  (b) $\tilde u_j \tilde u_j^T$ is feasible for \eqref{eq:H_i};
  (c) $\tilde u_j$ is almost as smooth as $u_j$.

Define
$\tilde u_j=(I-\EstProjection_{j-1}) u_j/\norm{(I-\EstProjection_{j-1}) u_j}\,.$
We first check the validity of this definition.
\begin{align*}
  \norm{\EstProjection_{j-1}u_j}_2
  \le & \norm{(\EstProjection_{j-1}-\PopProjection_{j-1})u_j}_2+
  \norm{\PopProjection_{j-1}v_j}_2+\norm{\PopProjection_{j-1}(u_j-v_j)}_2
  \le \epsilon_{j-1}+c_0 p^{-1}\,.
\end{align*}
When $\epsilon_{j-1}$ is small and $p$ large, $(I-\EstProjection_{j-1}) u_j\neq 0$, and
\begin{align}
  \norm{\tilde u_j - u_j}_2=&\left\lVert\frac{(I-\EstProjection_{j-1}) u_j}{\norm{(I-\EstProjection_{j-1}) u_j}_2} - \frac{u_j}{\norm{u_j}_2}\right\rVert_2
\le  2\norm{(I-\EstProjection_{j-1}) u_j - u_j}_2
\le
2(\epsilon_{j-1}+c_0 p^{-1})
\,,\label{eq:tilde-u-j-prox}
\end{align}
where the last inequality holds when $p$ is large and the first inequality follows from an elementary fact that, for all $u$, $v$,
$$
\left\lVert \frac{u}{\norm{u}_2} - \frac{v}{\norm{v}_2}  \right\rVert_2\le 2\frac{\norm{u-v}_2}{\max(\norm{u}_2,\norm{v}_2)}\,.
$$

Now we establish \eqref{eq:alpha=0-optimality-1}.
By feasibility of $\tilde u_j\tilde u_j^T$ we have
\begin{align}
  0\le & \innerp{\InputMatrix}{\PrimalVariable_j-\tilde u_j \tilde u_j^T} - \SmoothPenalty\innerp{\DiffOperator}{\PrimalVariable_j}
  +\SmoothPenalty\innerp{\DiffOperator}{\tilde u_j\tilde u_j^T} -\SparsePenalty(\norm{\PrimalVariable_j}_{1,1}
  -\norm{\tilde u_j \tilde u_j^T}_{1,1})\nonumber\\
  \le & \innerp{\InputMatrix}{\PrimalVariable_j-u_j u_j^T}
   - \SmoothPenalty\innerp{\DiffOperator}{\PrimalVariable_j}
  +\SmoothPenalty\innerp{\DiffOperator}{\tilde u_j\tilde u_j^T} +\SparsePenalty p +|\innerp{\InputMatrix}{\tilde u_j\tilde u_j^T-u_j u_j^T}|\,.\label{eq:optimality-1a}
\end{align}

We first bound $|\innerp{\InputMatrix}{\tilde u_j\tilde u_j^T-u_j u_j^T}|$:
\begin{align*}
  |\innerp{\InputMatrix}{\tilde u_j\tilde u_j^T-u_j u_j^T}|
  \le &
  \norm{\InputMatrix}_F\norm{\tilde u_j\tilde u_j^T-u_j u_j^T}_F
  \le (\norm{\PopMatrix}_F+\norm{\InputMatrix-\PopMatrix}_F)\norm{\tilde u_j\tilde u_j^T-u_j u_j^T}_F\\
  \le & 2\sqrt{2}\left(c_2 + e_n\right)p(\epsilon_j+c_0p^{-1})\,,
\end{align*}
where the last step uses \eqref{eq:tilde-u-j-prox}, \Cref{lem:popmatrix},
and the fact that $\norm{\InputMatrix-\PopMatrix}_{\infty,\infty}=e_n$.

Now we control $\SmoothPenalty\innerp{\DiffOperator}{\tilde u_j \tilde u_j^T}$.
When $\epsilon_{j-1}$ and $p^{-1}$ are small enough such that $\norm{(I-\EstProjection_{j-1})u_j}_2\ge 1/\sqrt{2}$,
we have
\begin{align*}
&\SmoothPenalty \innerp{\DiffOperator}{\tilde u_j \tilde u_j^T}=\SmoothPenalty  \norm{\DiffOperatorHalf \tilde u_j}_2^2 \le  2\SmoothPenalty\norm{\DiffOperatorHalf(I-\EstProjection_{j-1})u_j}_2^2\\
\le &  4\SmoothPenalty\left[\norm{\DiffOperatorHalf u_j}_2^2 + \left(\sum_{i=1}^{j-1}\norm{\DiffOperatorHalf \hat v_i}_2|\innerp{\hat v_i}{u_j}|\right)^2\right]
\le  4\left[2\SmoothPenalty L^2 p^{-4}+
\sum_{i=1}^{j-1}\beta_i
(\epsilon_{j-1}+c_0p^{-1})^2\right]\,,
\end{align*}
where the first two inequalities follow from multiple applications of
Cauchy-Schwartz, and the last inequality holds by definition of $\beta_i$, the smoothness
of $u_j$,
and the fact that $\sum_{i=1}^{j-1}|\innerp{\hat v_i}{u_j}|^2=\norm{\EstProjection_{j-1}u_j}_2^2$.
As a consequence, we establish \eqref{eq:alpha=0-optimality-1} from \eqref{eq:optimality-1a}
with
\begin{align*}
R_j=&2\sqrt{2}\left(c_2 +e_n \right)(\epsilon_{j-1}p+c_0)+\SparsePenalty p +
8\left[\SmoothPenalty L^2 p^{-4}+
\sum_{i=1}^{j-1}\beta_i
(\epsilon_{j-1}+c_0p^{-1})^2\right]=o_p(p)\,.\qedhere
\end{align*}
\end{proof}

\begin{proof}[\bf Proof of \Cref{thm:rate}]
  From assumption A1 and \Cref{lem:u-v} it suffices to prove that if $\norm{\InputMatrix-\PopMatrix}_{\infty,\infty}=O(e_n)$ then
$\sup_{1\le j\le k}\norm{\hat v_j-v_j}_2
     =O(e_n+\SparsePenalty)$.  

Consider the estimation procedure given by \eqref{opt2} for $j=1$.
Let $\mathbb B_{1}$ be the collection of all $p\times p$ symmetric matrices
with entries in $[-1,1]$. 
The optimization problem can be written in the following equivalent form.
\begin{align*}
  \max_{\PrimalVariable\in \Fantope{1}}
  \min_{\DualVariable\in\mathbb B_1} 
  \innerp{\InputMatrix}{\PrimalVariable} -\SparsePenalty \innerp{\DualVariable}{\PrimalVariable}\,.
\end{align*}
Let $\PrimalVariable^*$ be any maximizer, then 
$$
\PrimalVariable^*=\arg\max_{\PrimalVariable\in\Fantope{1}} \innerp{\InputMatrix - \SparsePenalty\DualVariable^*}{\PrimalVariable}=\arg\max_{\PrimalVariable\in\Fantope{1}} \innerp{\PopMatrix+\InputError - \SparsePenalty\DualVariable^*}{\PrimalVariable}
$$
where $\InputError=\InputMatrix-\PopMatrix$ and  $\DualVariable^*$ is the corresponding optimal dual variable.

By \Cref{lem:u-v} the eigengap of $\PopMatrix$ is of order
at least $p$ while the operator norm of $\InputError-\SparsePenalty\DualVariable^*$ is at most
$p(\norm{\InputError}_{\infty,\infty}+\SparsePenalty)$ which is $o(p)$. Thus
applying standard spectral subspace perturbation theory we know that
$H^*=\hat v_1\hat v_1^T$ where $\hat v_1$ is the leading eigenvector of
$\PopMatrix+\InputError - \SparsePenalty\DualVariable^*$, and satisfies for some constant $c$
$$
\norm{\hat v_1-v_1}_2\le c\norm{\InputError - \SparsePenalty\DualVariable^*}_F / p\le c(e_n+\SparsePenalty)\,.
$$

For $j=2,...,k$, we use induction.  Suppose that for $j-1$
we have $\norm{\hat v_{j-1} - v_{j-1}}_2$ and $\norm{\EstProjection_{j-1}-\PopProjection_{j-1}}_F$
are bounded by $O(e_n+\SparsePenalty)$.

Now consider the procedure \eqref{opt2} for $j$. Similarly let $\PrimalVariable^*$ be any solution and $\DualVariable^*$ the corresponding optimal dual variable.  We have
\begin{align*}
  \PrimalVariable^*=&\arg\max_{\PrimalVariable\in\D{\hat\Pi_{j-1}}}
   \innerp{\InputMatrix-\SparsePenalty
   \DualVariable^*}{\PrimalVariable}\\
   =&
   \arg\max_{\PrimalVariable\in\Fantope{1}}\innerp{
   (I-\hat\Pi_{j-1})(\InputMatrix-\SparsePenalty
   \DualVariable^*)(I-\hat\Pi_{j-1})}{\PrimalVariable}\,.
\end{align*}
The remainder of the proof focuses on analyzing the matrix $(I-\hat\Pi_{j-1})(\InputMatrix-\SparsePenalty
   \DualVariable^*)(I-\hat\Pi_{j-1})$.  In particular, we show that its leading eigenvector
   is close to $v_j$ with the desired rate.
We first write this matrix in four terms
\begin{align*}
  &(I-\hat\Pi_{j-1})(\PopMatrix+\InputError-\SparsePenalty
     \DualVariable^*)(I-\hat\Pi_{j-1})\\
     =&(I-\Pi_{j-1})\PopMatrix(I-\Pi_{j-1})\\
     &+(\Pi_{j-1}-\hat\Pi_{j-1})\PopMatrix(I-\hat\Pi_{j-1})+(I-\Pi_{j-1})\PopMatrix(\Pi_{j-1}-
     \hat\Pi_{j-1})\\
     & +(I-\hat \Pi_{j-1})(\InputError+\SparsePenalty \DualVariable^*)(I-\hat \Pi_{j-1})\\
     =& T_0+T_1+T_2\,.
\end{align*}
The main term is $T_0$.  The leading eigenvector of $T_0$ is $v_j$ with an eigengap  at least $\delta p / 2$ according to \Cref{lem:popmatrix}. Next we bound $T_1$ and $T_2$.
 In fact we have
\begin{equation}\label{eq:induction-error}
\norm{T_1}_F
     \le 2\norm{\Pi_{j-1}-\hat\Pi_{j-1}}_F\norm{\PopMatrix}_2 \le 2c_2 \norm{\Pi_{j-1}-\hat\Pi_{j-1}}_F p\,,  
\end{equation}
where $c_2$ is the constant in \Cref{lem:popmatrix}
and
$$
\norm{T_2}_F\le \norm{\InputError+\SparsePenalty \DualVariable^*}_F\le (e_n+\SparsePenalty)p\,.
$$
Then we have
\begin{equation}\label{eq:induction-error-2}
\norm{T_1+T_2}_F\le 2c_2\norm{\Pi_{j-1}-\hat\Pi_{j-1}}_F p +(e_n+\SparsePenalty)p\,.
\end{equation}
When $n$ and $p$ are large enough, $\norm{T_1+T_2}_F$ is smaller than the gap
between the first and second largest eigenvalues of $T_0$. Therefore, the
induction completes by using Davis-Kahan $\sin
\Theta$ theorem \citep[Theorem VII.3.1]{Bhatia97}
\begin{equation}\label{eq:induction-error-3}
  \norm{\hat v_j\hat v_j^T - v_jv_j^T}_F\le 2 \norm{T_1+T_2}_F/(\delta p / 2)\le 8c_2\delta^{-1} \norm{\Pi_{j-1}-\hat\Pi_{j-1}}_F+ 4\delta^{-1}(e_n+\SparsePenalty)\,,
\end{equation}
where $c_2$ is the constant given in \Cref{lem:popmatrix}.
\end{proof}

\noindent{\bf Proof of Technical Lemmas}
 \begin{proof}[\bf Proof of \Cref{lem:u-v}]
Note that $\tilde\PopOperator$ is a compact self-adjoint operator
from $L^2(0,1)$ to $L^2(0,1)$ with eigen-decomposition $ \tilde\PopOperator(s,t) = \sum_{j=1}^p \tilde\lambda_{j} \tilde\phi_{j}(s)\tilde\phi_{j}(t)\,$.

The Lipschitz condition on $\PopOperator$ implies that
\begin{align}\label{eq:HS-dist}
  \norm{\tilde\PopOperator-\PopOperator}_{\rm HS}^2\coloneqq \int\int |\PopOperator_{p}(s,t)-\PopOperator(s,t)|^2 dsdt\le \frac{L^2}{4p^2}\,.
\end{align}
By Weyl's inequality, $|\tilde\lambda_j-\lambda_j| \le\delta/2$
for large $p$.
Let $E_j$ and $\tilde E_{j}$ be the projection operators
onto the one-dimensional subspaces spanned by $\phi_j$ and $\tilde\phi_{j}$, respectively.
Then
\begin{align}
  \norm{\tilde\phi_{j}-\phi_j}_2\le \sqrt{2}\norm{\tilde E_{j}-E_{j}}_{\rm HS}\le \frac{4 \norm{\tilde\PopOperator-\PopOperator}_{\rm HS}}{\Gap}\le \frac{2L}{\Gap p}\,,
  \label{eq:phi-pj-phi-j}
\end{align}
where the first inequality follows from \Cref{lem:simple}, and the second inequality follows from the Davis-Kahan $\sin\Theta$ theorem (Chapter VII of \cite{Bhatia97}).

On the other hand, by assumption (A4) we have
\begin{align}
  \norm{\phi^*_{j}-\phi_j}_2\le\norm{\phi^*_{j}-\phi_j}_\infty\le \frac{L}{2p}\,.
  \label{eq:phi-pjs-phi-j}
\end{align}
which, together with \eqref{eq:phi-pj-phi-j}, implies that
\begin{align*}
  \norm{v_{j}-u_{j}^*}_2\le\left(\frac{2}{\Gap}+\frac{1}{2}\right)\frac{L}{p}\,.
\end{align*}
Also note that
\begin{align*}
  \norm{u_j^*-u_j}_2
    = &
    \left|\norm{u_j^*}_2-1\right|
    \le
    \norm{\phi_j^*-\phi_j}_2
    \le
    \norm{\phi_j^*-\phi_j}_\infty
    \le
      \frac{L}{2p}\,.
\end{align*}
Combining the previous two inequalities, we have
\begin{align*}
  \norm{v_j-u_j}_2\le & \left(\frac{2}{\Gap}+1\right)\frac{L}{p}\coloneqq c_0p^{-1}\,.\qedhere
\end{align*}
\end{proof}
\begin{proof}[\bf Proof of \Cref{lem:popmatrix}]
  The first claim follows from, letting $\PopMatrix^*$ be the
  discretized $\PopOperator$ evaluated at the grid,
  \begin{align*}
    \norm{\PopMatrix}_{F}^2 \le  &
    2\norm{a I_p}_F^2 + 2\norm{\PopMatrix^*}_F^2= 2 a^2 p + 2p^2 \norm{\tilde\PopOperator}_{\rm HS}^2\\ \le&
     2 a^2 p + 2p^2\left(2\norm{\PopOperator}_{\rm HS}^2 + 2\norm{\tilde\PopOperator-\PopOperator}_{\rm HS}^2\right)\le 2 a^2 p +2p^2\left(2\norm{\PopOperator}_{\rm HS}^2+\frac{L^2}{2p^2}\right)\\
    \le & p^2\left(2 a^2 p ^{-1}+4\norm{\PopOperator}_{\rm HS}^2+L^2\right)\le c_2^2 p^2\,,
  \end{align*}
  where \eqref{eq:HS-dist} is used to bound $\norm{\tilde\PopOperator-\PopOperator}_{\rm HS}$.

  The second claim follows from the fact that
  the eigengaps of $\PopMatrix$ are the same as those of
  $\PopMatrix^*$, and by Weyl's inequality:
  \begin{align*}
  \tilde\lambda_j - \tilde\lambda_{j+1}\ge & \lambda_j-\lambda_{j+1}-2\norm{\tilde \PopOperator - \PopOperator}_{\rm HS}\ge \delta - \frac{\delta}{2}=\delta/2\,.  \qedhere
  \end{align*}
\end{proof}

 \begin{proof}[\bf Proof of \Cref{lem:approx-curvature}]
Note that $v_j$ is the leading eigenvector of $\PopMatrix_j$, with eigengap at least $p\Gap/2$ as implied
by \Cref{lem:popmatrix}.
Then we have
\begin{align*}
  \norm{\PrimalVariable_j-u_j u_j^T}_F^2\le&  2\norm{\PrimalVariable_j-v_jv_j^T}_F^2 + 2\norm{v_jv_j^T - u_j u_j^T}_F^2
\le  \frac{8}{p\Gap}\innerp{-\PopMatrix}{\PrimalVariable_j-v_j v_j}+ 4c_0^2 p^{-2}\\
\le& \frac{8}{p\Gap}\innerp{-\PopMatrix}{\PrimalVariable_j- u_j u_j^T} + \frac{8}{p\Gap}
\norm{\PopMatrix}_F\norm{u_j u_j ^T - v_j v_j^T}_F + 4c_0^2 p^{-2}\\
\le & \frac{8}{p\Gap}\innerp{-\PopMatrix}{\PrimalVariable_j- u_j u_j^T} + \frac{8\sqrt{2}c_0c_2}{p\Gap}+4c_0^2 p^{-2}\\
\le &\frac{8}{p\Gap}\innerp{-\PopMatrix}{\PrimalVariable_j- u_j u_j^T} + \frac{12c_0c_2}{p\Gap}\,,
\end{align*}
where the first and third inequalities come from Cauchy-Schwartz, the second from the curvature
lemma of principal subspace (\cite{vu2013minimax}, Lemma 4.2), and the last holds provided that $p$ is sufficiently large.
\end{proof}

\bibliographystyle{apa-good}
\bibliography{sfpca}

\begin{thebibliography}{34}
\expandafter\ifx\csname natexlab\endcsname\relax\def\natexlab#1{#1}\fi
\expandafter\ifx\csname url\endcsname\relax
  \def\url#1{{\tt #1}}\fi
\expandafter\ifx\csname urlprefix\endcsname\relax\def\urlprefix{URL }\fi

\bibitem[{Bhatia(1997)}]{Bhatia97}
Bhatia, R. (1997).
\newblock {\em Matrix analysis\/}, vol. 169.
\newblock Springer.

\bibitem[{Boyd et~al.(2011)Boyd, Parikh, Chu, Peleato, \& Eckstein}]{Boyd-ADMM}
Boyd, S., Parikh, N., Chu, E., Peleato, B., \& Eckstein, J. (2011).
\newblock Distributed optimization and statistical learning via the alternating
  direction method of multipliers.
\newblock {\em Foundations and Trends{\textregistered} in Machine Learning\/},
  {\em 3\/}(1), 1--122.

\bibitem[{Bunea \& Xiao(2015)}]{bunea2014sample}
Bunea, F., \& Xiao, L. (2015).
\newblock On the sample covariance matrix estimator of reduced effective rank
  population matrices, with applications to fpca.
\newblock {\em Bernoulli\/}, {\em to appear\/}.

\bibitem[{Cardot(2000)}]{cardot2000nonparametric}
Cardot, H. (2000).
\newblock Nonparametric estimation of smoothed principal components analysis of
  sampled noisy functions.
\newblock {\em Journal of Nonparametric Statistics\/}, {\em 12\/}(4), 503--538.

\bibitem[{Castro et~al.(1986)Castro, Lawton, \&
  Sylvestre}]{castro1986principal}
Castro, P.~E., Lawton, W.~H., \& Sylvestre, E.~A. (1986).
\newblock Principal modes of variation for processes with continuous sample
  curves.
\newblock {\em Technometrics\/}, {\em 28\/}, 329--337.

\bibitem[{Chen \& M{\"u}ller(2012)}]{chen2012modeling}
Chen, K., \& M{\"u}ller, H.-G. (2012).
\newblock Modeling repeated functional observations.
\newblock {\em Journal of the American Statistical Association\/}, {\em
  107\/}(500), 1599--1609.

\bibitem[{Chiou \& M{\"u}ller(2009)}]{chiou2009modeling}
Chiou, J.-M., \& M{\"u}ller, H.-G. (2009).
\newblock Modeling hazard rates as functional data for the analysis of cohort
  lifetables and mortality forecasting.
\newblock {\em Journal of the American Statistical Association\/}, {\em
  104\/}(486), 572--585.

\bibitem[{{d'Aspremont} et~al.(2007){d'Aspremont}, {El Ghaoui}, Jordan, \&
  Lanckriet}]{AEJL:2007}
{d'Aspremont}, A., {El Ghaoui}, L., Jordan, M., \& Lanckriet, G. (2007).
\newblock A direct formulation of sparse {PCA} using semidefinite programming.
\newblock {\em SIAM Review\/}, {\em 49\/}(3).

\bibitem[{Dattorro(2005)}]{Dattorro}
Dattorro, J. (2005).
\newblock {\em Convex Optimization \& Euclidean Distance Geometry\/}.
\newblock Meboo Publishing USA.
\newblock V2012.01.28.

\bibitem[{Gasser et~al.(1985)Gasser, M{\"u}ller, K{\"o}hler, Prader, Largo, \&
  Molinari}]{gasser1985analysis}
Gasser, T., M{\"u}ller, H.-G., K{\"o}hler, W., Prader, A., Largo, R., \&
  Molinari, L. (1985).
\newblock An analysis of the mid-growth and adolescent spurts of height based
  on acceleration.
\newblock {\em Annals of Human Biology\/}, {\em 12\/}(2), 129--148.

\bibitem[{Hall \& Hosseini-Nasab(2006)}]{hall2006properties}
Hall, P., \& Hosseini-Nasab, M. (2006).
\newblock On properties of functional principal components analysis.
\newblock {\em Journal of the Royal Statistical Society: Series B (Statistical
  Methodology)\/}, {\em 68\/}(1), 109--126.

\bibitem[{Hall et~al.(2006)Hall, M{\"u}ller, \& Wang}]{hall2006properties2}
Hall, P., M{\"u}ller, H.-G., \& Wang, J.-L. (2006).
\newblock Properties of principal component methods for functional and
  longitudinal data analysis.
\newblock {\em The Annals of Statistics\/}, (pp. 1493--1517).

\bibitem[{Huang et~al.(2008)Huang, Shen, Buja et~al.}]{huang2008functional}
Huang, J.~Z., Shen, H., Buja, A., et~al. (2008).
\newblock Functional principal components analysis via penalized rank one
  approximation.
\newblock {\em Electronic Journal of Statistics\/}, {\em 2\/}, 678--695.

\bibitem[{Hyndman et~al.(2007)Hyndman, Ullah et~al.}]{hyndman2007robust}
Hyndman, R.~J., Ullah, S., et~al. (2007).
\newblock Robust forecasting of mortality and fertility rates: a functional
  data approach.
\newblock {\em Computational Statistics \& Data Analysis\/}, {\em 51\/}(10),
  4942--4956.

\bibitem[{James et~al.(2009)James, Wang, \& Zhu}]{james2009functional}
James, G.~M., Wang, J., \& Zhu, J. (2009).
\newblock Functional linear regression that's interpretable.
\newblock {\em The Annals of Statistics\/}, (pp. 2083--2108).

\bibitem[{Kneip et~al.(2011)Kneip, Sarda et~al.}]{kneip2011factor}
Kneip, A., Sarda, P., et~al. (2011).
\newblock Factor models and variable selection in high-dimensional regression
  analysis.
\newblock {\em The Annals of Statistics\/}, {\em 39\/}(5), 2410--2447.

\bibitem[{Lei \& Vu(2015)}]{lei2014sparsistency}
Lei, J., \& Vu, V.~Q. (2015).
\newblock Sparsistency and agnostic inference in sparse pca.
\newblock {\em The Annals of Statistics\/}, {\em to appear\/}.

\bibitem[{Li et~al.(2010)Li, Hsing et~al.}]{li2010uniform}
Li, Y., Hsing, T., et~al. (2010).
\newblock Uniform convergence rates for nonparametric regression and principal
  component analysis in functional/longitudinal data.
\newblock {\em The Annals of Statistics\/}, {\em 38\/}(6), 3321--3351.

\bibitem[{Lin(2013)}]{lin2013some}
Lin, Z. (2013).
\newblock Some perspectives of smooth and locally sparse estimators.
\newblock {\em Master thesis, Simon Fraser University, Canada\/}.

\bibitem[{Mackey(2008)}]{Mackey08}
Mackey, L. (2008).
\newblock Deflation methods for sparse pca.
\newblock In {\em NIPS\/}, vol.~21, (pp. 1017--1024).

\bibitem[{M{\"u}hl et~al.(1991)M{\"u}hl, Herkner, \& Swoboda}]{muhl1991mid}
M{\"u}hl, A., Herkner, K., \& Swoboda, W. (1991).
\newblock [the mid-growth spurt--a pre-puberty growth spurt. review of its
  significance and biological correlations].
\newblock {\em Padiatrie und Padologie\/}, {\em 27\/}(5), 119--123.

\bibitem[{Ramsay \& Silverman(2005)}]{ramsay2005functional}
Ramsay, J.~O., \& Silverman, B.~W. (2005).
\newblock {\em Functional {D}ata {A}nalysis\/}.
\newblock Springer Series in Statistics. New York: Springer, second ed.

\bibitem[{Rice \& Silverman(1991)}]{rice1991estimating}
Rice, J.~A., \& Silverman, B.~W. (1991).
\newblock Estimating the mean and covariance structure nonparametrically when
  the data are curves.
\newblock {\em Journal of the Royal Statistical Society. Series B
  (Methodological)\/}, (pp. 233--243).

\bibitem[{Sheehy et~al.(1999)Sheehy, Gasser, Molinari, \&
  Largo}]{sheehy1999analysis}
Sheehy, A., Gasser, T., Molinari, L., \& Largo, R. (1999).
\newblock An analysis of variance of the pubertal and midgrowth spurts for
  length and width.
\newblock {\em Annals of human biology\/}, {\em 26\/}(4), 309--331.

\bibitem[{Silverman(1996)}]{silverman1996smoothed}
Silverman, B.~W. (1996).
\newblock Smoothed functional principal components analysis by choice of norm.
\newblock {\em The Annals of Statistics\/}, {\em 24\/}(1), 1--24.

\bibitem[{Tuddenham \& Snyder(1954)}]{tuddenham1954physical}
Tuddenham, R., \& Snyder, M. (1954).
\newblock Physical growth of \textsc{C}alifornia boys and girls from birth to
  age 18.
\newblock {\em Calif. Publ. Child Deve.\/}, {\em 1\/}, 183--364.

\bibitem[{Van Der~Vaart \& Wellner(1996)}]{VanW96}
Van Der~Vaart, A.~W., \& Wellner, J.~A. (1996).
\newblock {\em Weak Convergence\/}.
\newblock Springer.

\bibitem[{Vu et~al.(2013)Vu, Cho, Lei, \& Rohe}]{vu2013fantope}
Vu, V.~Q., Cho, J., Lei, J., \& Rohe, K. (2013).
\newblock Fantope projection and selection: A near-optimal convex relaxation of
  sparse pca.
\newblock In {\em Advances in Neural Information Processing Systems\/}, (pp.
  2670--2678).

\bibitem[{Vu \& Lei(2013)}]{vu2013minimax}
Vu, V.~Q., \& Lei, J. (2013).
\newblock Minimax sparse principal subspace estimation in high dimensions.
\newblock {\em The Annals of Statistics\/}, {\em 41\/}(6), 2905--2947.

\bibitem[{White(1958)}]{white1958computation}
White, P.~A. (1958).
\newblock The computation of eigenvalues and eigenvectors of a matrix.
\newblock {\em Journal of the Society for Industrial \& Applied Mathematics\/},
  {\em 6\/}(4), 393--437.

\bibitem[{Wilmoth et~al.(2007)Wilmoth, Andreev, Jdanov, Glei, Boe, Bubenheim,
  Philipov, Shkolnikov, \& Vachon}]{wilmoth2007methods}
Wilmoth, J.~R., Andreev, K., Jdanov, D., Glei, D.~A., Boe, C., Bubenheim, M.,
  Philipov, D., Shkolnikov, V., \& Vachon, P. (2007).
\newblock Methods protocol for the human mortality database.
\newblock {\em University of California, Berkeley, and Max Planck Institute for
  Demographic Research, Rostock. URL: http://mortality. org [version
  31/05/2007]\/}.

\bibitem[{Yao et~al.(2005)Yao, M{\"u}ller, \& Wang}]{yao2005functional}
Yao, F., M{\"u}ller, H.-G., \& Wang, J.-L. (2005).
\newblock Functional data analysis for sparse longitudinal data.
\newblock {\em Journal of the American Statistical Association\/}, {\em
  100\/}(470), 577--590.

\bibitem[{Zhao et~al.(2012)Zhao, Ogden, \& Reiss}]{zhao2012wavelet}
Zhao, Y., Ogden, R.~T., \& Reiss, P.~T. (2012).
\newblock Wavelet-based lasso in functional linear regression.
\newblock {\em Journal of Computational and Graphical Statistics\/}, {\em
  21\/}(3), 600--617.

\bibitem[{Zhou et~al.(2013)Zhou, Wang, \& Wang}]{zhou2013functional}
Zhou, J., Wang, N.-Y., \& Wang, N. (2013).
\newblock Functional linear model with zero-value coefficient function at
  sub-regions.
\newblock {\em Statistica Sinica\/}, {\em 23\/}(1), 25.

\end{thebibliography}

%


\end{document}